\documentclass[11pt]{article}
\usepackage[utf8]{inputenc}
\usepackage{graphicx}
\usepackage[paper = a4paper, margin= 1in]{geometry}
\usepackage{amsmath}
\usepackage[makeroom]{cancel}
\usepackage{float}
\usepackage{amssymb}
\usepackage{amsthm}
\usepackage{titlesec}
\usepackage{enumerate}
\usepackage{enumitem}
\usepackage[usestackEOL]{stackengine}
\usepackage{algorithm}
\usepackage[noend]{algpseudocode}
\usepackage{flexisym}
\usepackage[algo2e,norelsize,linesnumbered,ruled]{algorithm2e}
\newtheorem{theorem}{\protect Theorem}
\newtheorem{lemma}{\protect Lemma}
\newtheorem*{pf}{\protect Proof}
\newtheorem{assume}{\protect Assumption}
\newtheorem{defn}{Definition}
\usepackage{comment}
\usepackage{soul,color} 
\usepackage{tikz}
\usepackage{cite}
\usepackage{authblk}
\title{Hierarchical MPC for coupled subsystems using adjustable tubes}
\author[]{Vignesh Raghuraman}
\author[]{Justin P. Koeln}
\affil[]{\textit{University of Texas at Dallas, Richardson, TX, 75080, United States} \\ (e-mail: vignesh.raghuraman@utdallas.edu, justin.koeln@utdallas.edu) \vspace{4mm} \\ \textit{Preprint submitted for Elsevier}}
\date{}                     
\usepackage{fancyhdr}
\begin{document}
\maketitle

\noindent\rule{16.2cm}{0.4pt}
\\ \textbf{Abstract}
\\ A hierarchical Model Predictive Control (MPC) formulation is presented for coupled discrete-time linear systems with state and input constraints. Compared to a centralized approach, a two-level hierarchical controller, with one controller in the upper-level and one controller per subsystem in the lower-level, can significantly reduce the computational cost associated with MPC. Hierarchical coordination is achieved using adjustable tubes, which are optimized by the upper-level controller and bound permissible lower-level controller deviations from the system trajectories determined by the upper-level controller. The size of these adjustable tubes determines the degree of uncertainty between subsystems and directly affects the required constraint tightening under a tube-based robust MPC framework. Sets are represented as zonotopes to enable the ability to optimize the size of these adjustable tubes and perform the necessary constraint tightening online as part of the MPC optimization problems. State and input constraint satisfaction is proven for the two-level hierarchical controller with an arbitrary number of controllers at the lower-level and a numerical example demonstrates the key features and performance of the approach.
\vspace{8mm}
\\ \noindent\rule{16.2cm}{0.4pt}

\section{Introduction}

Model Predictive Control (MPC) of constrained dynamic systems provides the ability to satisfy both input and state constraints to guarantee safe and reliable system operation. This is particularly important for systems where the desired operation requires both transient and steady-state input and state trajectories to approach these constraints. Examples include the control of water distribution networks \cite{Wang2017a}, aircraft power systems \cite{Seok2017a}, smart power grids \cite{Khan2016,Irfan2017}, and hybrid electric vehicles \cite{Enang2017,Wang2020_DSCC}. 
However, centralized MPC approaches are not well-suited for the control of these complex multi-timescale systems, where the system is comprised of many dynamically coupled subsystems and achieving the desired operation requires both fast control update rates and long prediction horizons. 

For these complex systems, hierarchical MPC can be used to decompose control decision across multiple levels of controllers \cite{Scattolini2009}. Typically, upper-level controllers 
are designed with large time step sizes to optimize system operation over long prediction horizons while lower-level controllers use small time step sizes to resolve the fast dynamics of the system over short prediction horizons.
With a single controller per level, \emph{vertical} hierarchical MPC is a computationally efficient approach for controlling multi-timescale systems with a relatively low number of states and inputs \cite{Koeln2019_Aut}. For more complex systems, comprised of multiple dynamically-coupled subsystems, \emph{full} hierarchical MPC utilizes multiple controllers at each of the lower-levels to reduce the number of control decisions per controller 
\cite{Barcelli2010,Farina2017b,Farina2018}.


Many hierarchical MPC formulations \cite{Farina2018,Farina2017b,Barcelli2010,Vermillion2014} have been developed with a focus on either practical application or theoretical guarantees. Specifically, the two-level hierarchical controller 
in \cite{Barcelli2010} with 
an upper-level MPC and a lower-level linear controller 
achieves state and input constraint satisfaction through communication of optimal references and reference rate changes between controller levels and also guarantees closed-loop stability. The controller developed in \cite{Vermillion2014} provides guaranteed persistent controller feasibility and closed-loop stability for cascaded system with actuator dynamics subject to input constraints. 
In this case, 
coordination 
is achieved through the appropriate choice of contractive terminal constraint sets and terminal control laws, which overall guarantee stability of the error between inner-loop and outer-loop reference models to the origin. The works in \cite{Farina2017b,Farina2018} extend the vertical hierarchical architecture to a full two-level hierarchical controller with one upper-level controller and multiple controllers at the lower-level, one for each subsystem operating at the same timescale in \cite{Farina2017b} and different timescales in \cite{Farina2018},
and guarantees closed-loop stability and input constraint satisfaction while driving the system to a desired steady state. Note that \cite{Farina2018,Barcelli2010,Vermillion2014} focus on driving the system to a desired steady-state equilibrium. However, for systems with finite operation, such steady-state equilibrium might not exist, as in the case of systems whose operation is based on the utilization of a finite resource (e.g. battery state of charge in an electric vehicle \cite{Wang2020_DSCC,Sampathnarayanan2014} or fuel in an aircraft \cite{doman2015rapid}).


Similar to \cite{Richards2003,Koeln2019_Aut}, this work focuses on the notion of \emph{completion}, with the goal of maximizing transient performance by satisfying state, input, and terminal constraints during system operation. While the multi-rate hierarchical MPC proposed in \cite{Farina2018} achieves 
real-time computational performance using a \emph{full} hierarchical MPC architecture
with a reduced-order model at the upper-level, guarantees on closed-loop state constraint satisfaction are not shown explicitly.
Additionally, the amount of control flexibility provided to the upper- and lower-level controllers along with the resulting uncertainty sets, robust positive invariant (RPI) sets, and tightened constraint sets are determined offline and might not be the optimal choice for systems that need a time-varying control flexibility. Moreover, guaranteed convergence might not be possible for a wide range of systems due to underlying assumptions on the slow timescale of the upper-level controller.
To address these challenges, this work focuses on development of a set-based hierarchical MPC architecture for linear systems of dynamically-coupled subsystems that guarantees state and input constraint satisfaction. 


One of the fundamental considerations for coordination in hierarchical MPC is how to provide lower-level controllers the flexibility to use their fast update rates and the fast dynamics of the system to improve upon the control decisions made by upper-level controllers without introducing unnecessary conservatism to account for this flexibility. In the authors' prior work \cite{Koeln2019_Aut}, set-based vertical hierarchical MPC was proposed, where \emph{waysets} were used as the primary coordination mechanism to provide both control performance and guaranteed constraint satisfaction. 
Strategically designed terminal costs were added to complement the waysets to guarantee that the lower-level controllers can only improve control performance compared to the upper-level controller trajectories \cite{Raghuraman2020_ACC}.



For full hierarchical MPC of systems of dynamically-coupled subsystems, providing lower-level controllers the flexibility to deviate from the trajectories planned by upper-level controllers introduces uncertainty between subsystems. Therefore, the desired degree of flexibility balances the benefits of allowing lower-level controller to improve control performance within their own subsystem with the cost of creating unknown disturbances for neighboring subsystems. This trade off can be time-varying, where certain system operations might require a high level of coordination between subsystems, resulting in very little flexibility for lower-level subsystem controllers to deviate
from the upper-level system-wide control plan. Alternatively, other system operations might not require much coordination between subsystems and lower-level controllers should be permitted a high degree of flexibility to further improve control performance.   

The proposed two-level hierarchical MPC framework provides this time-varying subsystem coordination flexibility using an adjustable tube set-based coordination mechanism. Specifically, while planning system state and input trajectories, the upper-level controller simultaneously optimizes the permissible deviations from these trajectories provided to the lower-level subsystem controllers and the corresponding constraint tightening needed to be robust to these deviations. These time-varying permissible deviations are communicated to the lower-level controllers that use this flexibility to further optimize subsystem operation. The ability to embed the optimization of these permissible deviation bounds within the upper-level MPC optimization problem is enabled by \emph{zonotopes} 
and the recent work on computing Robust Positive Invariant (RPI) sets and Pontryagin difference set operations using linear constraints \cite{Raghuraman2019,Raghuraman_2021_ACC}. 

The specific contributions of this paper are: (1) the development a two-level hierarchical framework with $M$ lower-level controllers, one for each of the $M$ dynamically-coupled subsystems; (2) the definition and the use of adjustable tubes to provide time-varying bounds on permissible deviations between upper-level and lower-level planned trajectories; (3) the closed-loop analysis of the hierarchical controller to prove controller feasibility and guarantee constraint satisfaction; 
and (4) a numerical demonstration of the capabilities of the proposed approach. Note that the proposed work extends the tube-based robust MPC with uncertainty set optimization from \cite{Raghuraman_2021_ACC} to a hierarchical MPC framework with optimal allocation of uncertainty quantified as the differences in control decisions between controller levels and between subsystems.
Similar to \cite{Raghuraman_2021_ACC}, RPI, tightened output, and tightened terminal sets corresponding to the optimized uncertainty are computed online while solving the control optimization problem.

\section{Notation} \label{Notation}
For a system comprised of multiple subsystems, system-level vectors are denoted in bold, e.g. state $\mathbf{x}$ and input $\mathbf{u}$, while vectors of the $i^{th}$ subsystem have the subscript $ i $, e.g. state $ x_i $ and input $ u_i $. The system state vector is formed by the concatenation of subsystem state vectors as $ \mathbf{x} = [x_i] $. Alternatively, the states of subsystem $ i $ can be extracted from the system state vector as $x_i = \Pi_i \mathbf{x}$. For a discrete-time system, $\mathbf{x}(k)$ denotes the state $\mathbf{x}$ at time step $k$. With $ [k, k+N-1] $ denoting the integers from $ k $ to $ k+N-1 $, the input trajectory over these time steps is denoted $\{\mathbf{u}(j)\}_{j = k}^{k+N-1}$. For MPC, the double index notation $\mathbf{x}(k+l|k)$ denotes the predicted state at future time $k+l$ determined at time step $k$.  
The block-diagonal matrix $ K $ with blocks $K_i$ is denoted $ K = diag(K_i) $. The $p$-norm of a vector is denoted $||\cdot||_p$ and the weighted norm is $||\mathbf{x}||_{\Lambda}^2 = \mathbf{x}^T\Lambda\mathbf{x}$, where $\Lambda$ is a positive-definite diagonal matrix. 
All sets are shown in caligraphic font. For sets $\mathcal{X}$, $\mathcal{Y} \in \mathbb{R}^n$, $\mathcal{X} \oplus \mathcal{Y}$ denotes the Minkowski sum and $\mathcal{X} \ominus \mathcal{Y}$ denotes the Minkowski/Pontryagin difference of $\mathcal{Y}$ from $\mathcal{X}$. 
The Cartesian product of sets is denoted as $ \mathcal{X} \times \mathcal{Y} $. The projection of $\mathcal{X}$ on the $n_i$ dimensions of subsystem $i$ is denoted as $\mathcal{X}_i = \Pi_i\; \mathcal{X}$.

\section{Problem Formulation}
Consider a linear discrete time-invariant system composed of $ M $ dynamically-coupled subsystems, $\mathbf{S}_i$, where $ i \in \mathcal{N} \triangleq [1, M]$. The dynamics of subsystem $\mathbf{S}_i$ are
\begin{subequations} \label{subsys_defn}
\begin{align}
    x_i(k+1) &= A_{ii}x_i(k) + B_{ii}u_i(k) + w_i(k), \label{subsys_dynamics} \\
    y_i(k) &= C_{i}x_i(k) + D_{i}u_i(k), \label{subsys_outputs}
\end{align}
\end{subequations}
where $x_i \in \mathbb{R}^{n_i}$ are the states, $u_i \in \mathbb{R}^{m_i}
$ are the inputs, and $y_i \in \mathbb{R}^{n_i+m_i}$ are the outputs. The coupling between subsystems is captured by the disturbance vector 
\begin{equation}\label{disturbances}
    w_i(k) = \sum_{j \in \mathcal{N}_i} (A_{ij}x_j(k) + B_{ij}u_j(k)),
\end{equation}
where $\mathcal{N}_i$ is the set of neighboring subsystems such that
\begin{equation}
	\mathcal{N}_i \triangleq \{ j \in \mathcal{N} \setminus \{ i \} : [ A_{ij} \; B_{ij} ] \neq 0 \}.
\end{equation}
The outputs are defined to include all states and inputs such that $ y_i(k) \triangleq [x_i(k)^\top \; u_i(k)^\top]^\top $ and $[ C_i \; D_i ] \triangleq I_{n_i+m_i}$. The subsystem states, inputs, and outputs are constrained such that 
\begin{equation} \label{subsystem_cons}
    x_i(k) \in \mathcal{X}_i, \; u_i(k) \in \mathcal{U}_i, \; y_i(k) \in \mathcal{Y}_i \triangleq \mathcal{X}_i \times \mathcal{U}_i.
\end{equation}
Based on \eqref{subsys_defn} and \eqref{disturbances}, the full system dynamics are
\begin{subequations} \label{sys_defn}
\begin{align}
    \mathbf{x}(k+1) &= A\mathbf{x}(k) +B\mathbf{u}(k), \label{sys_dynamics} \\
    \mathbf{y}(k) &= C\mathbf{x}(k) + D\mathbf{u}(k), \label{sys_outputs}
\end{align}
\end{subequations}
where $\mathbf{x} = [x_i] \in \mathbb{R}^n$, $\mathbf{u} = [u_i] \in \mathbb{R}^m$, and $\mathbf{y} = [y_i] \in \mathbb{R}^{n+m}$, such that $ n = \sum_{i=1}^M n_i $ and $ m = \sum_{i=1}^M m_i $. 
The system constraints are
\begin{subequations} \label{sys_cons}
\begin{align}
    \mathbf{x}(k) \in \mathcal{X} &\triangleq \mathcal{X}_1 \times \dots \times \mathcal{X}_M, \\
    \mathbf{u}(k) \in \mathcal{U} &\triangleq \mathcal{U}_1 \times \dots \times \mathcal{U}_M, \\
    \mathbf{y}(k) \in \mathcal{Y} &\triangleq \mathcal{Y}_1 \times \dots \times \mathcal{Y}_M.
\end{align}
\end{subequations}
Let $ A_D \triangleq diag(A_{ii}) $ and  $ B_D \triangleq diag(B_{ii}) $ be block diagonal matrices while $ A_C \triangleq A - A_D $ and $ B_C \triangleq B - B_D $ are off-diagonal matrices that capture the coupling between subsystems.
\begin{assume}\label{assm_stabilizable}
There exists a static feedback control gain $ K_i \in \mathbb{R}^{m_i \times n_i} $ for each subsystem $ \mathbf{S}_i$, $ i \in \mathcal{N} $, such that $ A_{ii} + B_{ii} K_i $ is Schur stable and $ A + B K $ is Schur stable, where $ K = diag(K_i)$ is a block-diagonal matrix.
\end{assume}
\begin{assume}\label{asm_finiteop}
	With a fixed time step $\Delta t$, the system operates for a finite length of time starting from $t = 0$ and ending at $t = t_F = k_F \Delta t$ with time steps indexed by $k \in [0, k_F]$. 
\end{assume}

Starting from an initial condition $\mathbf{x}(0)$, the goal is to plan and execute an input trajectory and corresponding state and output trajectories satisfying the system dynamics from \eqref{sys_defn}, the constraints from \eqref{sys_cons} for all $ k \in [0, k_F-1] $, and the terminal constraint 
\begin{equation}\label{sys_Tcons}
 	\mathbf{x}(k_F) \in \mathcal{T} \triangleq \mathcal{T}_1 \times \dots \times \mathcal{T}_M \subseteq \mathcal{X}.
 \end{equation}
\begin{assume}
	The sets $\mathcal{X}_i$, $\mathcal{U}_i$, and $\mathcal{T}_i $, $i \in \mathcal{N}$, are zonotopes. 
\end{assume}
The generic cost function 
\begin{equation}\label{sys_cost}
	J(\mathbf{x}(0)) = \sum\limits_{j=0}^{k_F-1} \ell(j) + \ell_F(k_F),
\end{equation}
defines the cost of system operation using a pre-determined reference trajectory $ \{ \mathbf{r}(k) \}_{k =0}^{k_F}$ with stage costs $ \ell(j) = \ell( \mathbf{x}(j), \mathbf{u}(j), \mathbf{r}(j)) $ and terminal cost $ \ell_F(k_F) = \ell_F( \mathbf{x}(k_F), \mathbf{r}(k_F)) $. 

Considering the full system \eqref{sys_defn}, operational constraints \eqref{sys_cons}, terminal constraint \eqref{sys_Tcons}, and cost function \eqref{sys_cost}, this paper develops a two-level hierarchical control approach with $ M $ controllers at the lower-level that guarantees constraint satisfaction and provides computational efficiency in the case of a large number of subsystems $ M $, small time step size $\Delta t$, and large operating duration $ t_F $.

\section{Hierarchical Control}
The proposed hierarchical control formulation consists of a single controller $\mathbf{C}_0$ in the upper-level and $ M $ controllers $\mathbf{C}_i, i \in \mathcal{N}$, in the lower-level, where $ \mathbf{C}_i $ controls subsystem $\mathbf{S}_i$.

\begin{assume} \label{updateRates}
	The controller $ \mathbf{C}_0 $ has a time step size $ \Delta t_0 $ and maximum prediction horizon $ \bar{N}_0 $ such that $ \Delta t_0 \bar{N}_0 = t_F $. Each controller $\mathbf{C}_i $, $i \in \mathcal{N} $, has a time step size $ \Delta t $ and maximum prediction horizon $ \bar{N} $ such that $ \Delta t \bar{N} = \Delta t_0 $.
\end{assume}
Let $\nu_0 \triangleq \frac{\Delta t_0}{\Delta t} = \bar{N} \in \mathbb{Z}_{+}$ be defined as a time scaling factor for $\mathbf{C}_0$. The time steps for $\mathbf{C}_0$ are indexed by $k_0,$ with $k_0 \triangleq \frac{k}{\nu_0}$, and let $k_{0,F} \triangleq \frac{k_F}{\nu_0} = \bar{N}_0 $ denote the terminal step of $\textbf{C}_0$ such that $k_0 \in [0, k_{0,F}] $. Thus, the upper-level controller $ \mathbf{C}_0 $ has a \emph{shrinking horizon}, with time-varying horizon length $ N_0(k_0) \triangleq \bar{N}_0 - k_0 $. Each lower-level controller $\mathbf{C}_i$ has a \emph{shrinking and resetting horizon}, with horizon length $ N(k) \triangleq \bar{N} - (k \mod \bar{N})$. This allows $\mathbf{C}_i$ to predict between updates of $\mathbf{C}_0$, at which point ($k \mod \bar{N} = 0$) and the prediction horizon resets back to $N(k) = \bar{N}$.

Similar to \cite{Koeln2019_ACC,Koeln2019_Aut}, $\mathbf{C}_0$ predicts coarse state and input trajectories at time indices $k_0$ with a large time step size $\Delta t_0$. Lower-level controllers $\mathbf{C}_i$ are permitted bounded deviations from the trajectories planned by $\mathbf{C}_0$ to further improve control performance using a smaller time step size $\Delta t $. 
Unlike \cite{Koeln2019_ACC,Koeln2019_Aut}, this work addresses the coupling between subsystems. If the lower-level controller $\mathbf{C}_i $ chooses to deviate from the state and input trajectories planned by $\mathbf{C}_0$, these deviations create unknown disturbances that could lead to constraint violations in neighboring subsystems. Therefore, instead of using waysets as in \cite{Koeln2019_ACC,Koeln2019_Aut}, a tube-based coordination mechanism is used to bound the permissible deviations between the trajectories planned by $\mathbf{C}_0$ and those planned by $\mathbf{C}_i$. Moreover, the size of these permissible deviations is optimized online by $\mathbf{C}_0$ to balance the flexibility provided to lower-level controllers with the potentially time-varying need for close coordination among subsystems.

Specifically, for each subsystem, the sets $\Delta\mathcal{Z}_i(\delta_i^z(k_0))$ and $\Delta\mathcal{V}_i(\delta_i^v(k_0)) $ denote scaled zonotopes that bound the permissible state and input deviations between the trajectories planned by $\mathbf{C}_0$ and those planned by $\mathbf{C}_i$. The scaling vectors can be collected to form the output deviation vector $\delta_i(k_0) = [ \delta_i^z(k_0)^\top \; \delta_i^v(k_0)^\top ]^\top $ and the permissible output deviation set
\begin{equation}
    \Delta\mathcal{Y}_i(\delta_i(k_0)) = \Delta\mathcal{Z}_i(\delta_i^z(k_0)) \times \Delta\mathcal{V}_i(\delta_i^v(k_0)).
\end{equation}
To reduce notational complexity, the shorthand $\Delta\mathcal{Y}_i(k_0) = \Delta\mathcal{Y}_i(\delta_i(k_0)) $ is used when explicitly stating the dependency on $ \delta_i(k_0) $ is unnecessary. The system state, input, and output deviation vectors are $\boldsymbol{\delta}^z(k_0) = [\delta_i^z(k_0)]$, $\boldsymbol{\delta}^v(k_0) = [\delta_i^v(k_0)]$, and $\boldsymbol{\delta}(k_0) = [ \boldsymbol{\delta}^z(k_0)^\top \; \boldsymbol{\delta}_v(k_0)^\top ]^\top$ and the scaled subsystem deviation sets combine to form the scaled system deviation sets
\begin{subequations}\label{sys_perm_dev_sets_all}
\begin{align}
    \Delta\mathcal{Z}(\boldsymbol{\delta}(k_0)) &= \Delta\mathcal{Z}_1(k_0) \times \dots \times \Delta\mathcal{Z}_M(k_0), \label{sys_perm_dev_sets_DeltaZ} \\
    \Delta\mathcal{V}(\boldsymbol{\delta}(k_0)) &= \Delta\mathcal{V}_1(k_0) \times \dots \times \Delta\mathcal{V}_M(k_0), \label{sys_perm_dev_sets_DeltaV} \\
    \Delta\mathcal{Y}(\boldsymbol{\delta}(k_0)) &= \Delta\mathcal{Z}(\boldsymbol{\delta}(k_0)) \times \Delta\mathcal{V}(\boldsymbol{\delta}(k_0)).
    \end{align}
\end{subequations}
The controller $\textbf{C}_0$ updates only when $k = \nu_0 k_0$ (i.e. when $k \mod \nu_0 = 0$), by solving the constrained optimization problem $\mathbf{P}_0(\mathbf{x}(k))$ defined as 
\begin{subequations}
\begin{align}
	& J_0^*\left(\mathbf{x}(k)\right) = \mkern-15mu \min_{\text{\tiny $\begin{matrix} \hat{\mathbf{x}}(k_0|k_0), \hat{\mathbf{U}}(k_0), \\[-2pt] \boldsymbol{\delta}(k_0) \end{matrix} $}} \mkern-2mu	\sum_{j=k_0}^{k_{0,F}-1} \mkern-5mu \ell \left( j|k_0 \right) + \ell_F(k_{0,F}), \label{Up_cost} \\
	&\text{s.t.} \forall j \in \left[k_0,k_{0,F}-1\right],  
	\nonumber \\
	&\hat{\mathbf{x}}(j+1|k_0) = A_0 \hat{\mathbf{x}}(j|k_0) + B_0 \hat{\mathbf{u}}(j|k_0), \label{Up_model} \\
	& \hat{\mathbf{y}}(j|k_0) = C \hat{\mathbf{x}}(j|k_0) + D \hat{\mathbf{u}}(j|k_0) \in \hat{\mathcal{Y}}_0(\boldsymbol{\delta}(k_0)), \label{Up_outputs} \\
	&\hat{\mathbf{x}}(k_{0,F}|k_0) \in 
	\hat{\mathcal{T}}_0(\boldsymbol{\delta}(k_0)), \label{Up_terminal} \\
	& \mathbf{x}(k) - \hat{\mathbf{x}}(k_0|k_0) \in \Delta\mathcal{Z}(\boldsymbol{\delta}(k_0)) \oplus \mathcal{E}_0(\boldsymbol{\delta}(k_0)), \label{Up_IC} \\
	& \Delta\mathcal{Z}(\boldsymbol{\delta}(k_0)) \subseteq \text{Pre}(\Delta\mathcal{Z}(\boldsymbol{\delta}(k_0))). \label{Up_Pre}
\end{align} \label{Up_MPC}%
\end{subequations}
The shrinking horizon of  $\mathbf{P}_0(\mathbf{x}(k))$ is reflected in the summation limits in \eqref{Up_cost}. The stage costs are defined as $\ell(j|k_0) = \ell( \mathbf{x}(k), \hat{\mathbf{x}}(j|k_0), \hat{\mathbf{u}}(j|k_0), \boldsymbol{\delta}(k_0), \mathbf{r}_0(j))$ to be a function of the measured state, nominal state, nominal input, permissible deviations for lower-level controllers, and the reference trajectory. The terminal cost $\ell_F(k_F) $ is the same as in \eqref{sys_cost}. Note that the system performance can be balanced with the maximization of $\boldsymbol{\delta}$ through the addition of the term $\Lambda ||\bar{\boldsymbol{\delta}} - \boldsymbol{\delta}||_p$, where $\Lambda$ is a scalar weighting term and $\bar{\boldsymbol{\delta}}$ is a user-specified upper-bound on $\boldsymbol{\delta}$. The nominal input trajectory is defined as $\hat{\mathbf{U}}(k_0) = \{ \hat{\mathbf{u}}(j|k_0)\}_{j = k_0}^{k_{0,F} - 1}$. The permissible output deviation scaling vector $ \boldsymbol{\delta}(k_0) $ affects the sizes of the tightened output constraint set $ \hat{\mathcal{Y}}_0(\boldsymbol{\delta}(k_0)) $, the tightened terminal constraint set $ \hat{\mathcal{T}}_0(\boldsymbol{\delta}(k_0)) $, the state deviation constraint set $ \Delta\mathcal{Z}(\boldsymbol{\delta}(k_0)) $, and the RPI set $ \mathcal{E}_0(\boldsymbol{\delta}(k_0)) $, which are time-varying. 
In \eqref{Up_model}, the model used by $\mathbf{C}_0$ assumes a piecewise constant control input over the time step size $\Delta t_0$ and thus $A_0 = A^{\nu_0}$ and $B_0 = \sum_{j=0}^{\nu_0 -1} A^{j}B$ (as in \cite{Scattolini2007}). In \eqref{Up_outputs} and \eqref{Up_terminal}, the outputs and terminal state are constrained to the time-varying
tightened output and terminal constraint sets. Similar to tube-based MPC \cite{Mayne2005}, \eqref{Up_IC} allows $\mathbf{C}_0$ flexibility in the choice of initial condition $ \hat{\mathbf{x}}(k_0|k_0) $, which is used to prove recursive feasibility of $\mathbf{P}_0(\mathbf{x}(k))$ in Section \ref{sec_hier_control_feas}. Finally, \eqref{Up_Pre} constrains the time-varying permissible state deviation set to be a subset of its own precursor set. Based on the definition from \cite{Borrelli2011}, the precursor set is defined specifically as
\begin{equation} \label{precursor}
    \text{Pre}(\Delta\mathcal{Z}(k_0)) = \left\{\mathbf{z} \mid 
    \exists \, \mathbf{v} \in \Delta\mathcal{V}(k_0) \text{ s.t. }  \\ A_D \mathbf{z} + B_D \mathbf{v} \in \Delta\mathcal{Z}(k_0) \right\}, 
\end{equation}
and is used to establish feasibility of the lower-level controllers in Section \ref{sec_hier_control_feas}. The details of how to formulate the sets and set containment conditions used in \eqref{Up_MPC} 
as linear constraints are provided in the \textbf{Appendix}. The reference trajectory $\mathbf{r}_0(j)$ can be obtained by downsampling the predetermined reference trajectory $\mathbf{r}(j)$ either using averaging or zero order hold \cite{Koeln2019_ACC}.
Note that the RPI set $ \mathcal{E}_0(\boldsymbol{\delta} (k_0)) $ is assumed to be a \emph{structured} RPI set such that
\begin{equation}\label{Up_struct_RPI}
    \mathcal{E}_0(\boldsymbol{\delta}(k_0)) = \mathcal{E}_1(\delta_1(k_0)) \times \dots \times \mathcal{E}_M(\delta_M(k_0)), 
\end{equation}
and is formulated in more detail in Section \ref{sec_nom_traj_error_propagation}. 

The lower-level controllers $\mathbf{C}_i $, $ i \in \mathcal{N}$, update at each time index $ k $ by each solving, in parallel, the constrained optimization problems $\textbf{P}_i(x_i(k))$, defined as 
\begin{subequations}
	\begin{align}
	& J_i^*\left(x_i(k)\right) = \mkern-5mu \min_{\text{\tiny $ \begin{matrix} z_i(k|k), \\[0pt] V_i(k) \end{matrix} $}} \mkern-20mu \sum_{j=k}^{k+N(k)-1} \mkern-20mu \ell_i(j|k) + \ell_{i,F}(k+N(k)), \label{Low_cost} \\
	&\text{s.t.} \, \forall j \in \left[k,k+N(k)-1\right], \nonumber \\
	& z_i(j+1|k) = A_{ii}z_i(j|k) + B_{ii}v_i(j|k) + \hat{w}_i^*(j), \label{Low_model} \\
	& y_i(j|k) = C_i z_i(j|k) + D_i v_i(j|k), \label{Low_outputs} \\
    & y_i(j|k) - \hat{y}_i^*(j) \in \Delta\mathcal{Y}_i(\delta_i^*(k_0)), \label{Low_DeltaY_deviation} \\
    & z_i(k+N(k)|k) - \hat{x}_i^*(k+N(k)) \in \Delta\mathcal{Z}_i(\delta_i^*(k_0)), \label{Low_terminal} \\
    & x_i(k) - z_i(k|k) \in \mathcal{E}_i(\delta_i^*(k_0)). \label{Low_IC}
	\end{align}
	\label{Low_MPC}%
\end{subequations}
The shrinking and resetting horizon of  $\mathbf{P}_i(x_i(k))$ is reflected in the summation limits in \eqref{Low_cost}. The stage costs are defined as $ \ell_i(j|k) = \ell_i \left(x_i(k),z_i(j|k), v_i(j|k), r_i(j) \right) $ to be a function of the measured subsystem state, nominal subsystem state, nominal subsystem input, and subsystem reference trajectory. The terminal cost is defined as $ \ell_{i,F}(k+N(k)) = \ell_{i,F}( z_i(k+N(k)|k), r_i(k+N(k))) $. The nominal input trajectory is defined as $V_i(k) = \{ v_i(j|k)\}_{j = k}^{k +N(k) - 1}$. In \eqref{Low_model}, the subsystem dynamics from \eqref{subsys_dynamics} are used with a time-varying $\mathbf{C}_0$-optimal disturbance $\hat{w}_i^*(j)$ that is communicated from $\textbf{C}_0$ (details in Section \ref{sec_nom_traj_error_propagation}). Nominal subsystem outputs are defined in \eqref{Low_outputs} and the differences between these outputs and the $\mathbf{C}_0$-optimal outputs $\hat{y}_i^*(j)$ are constrained in \eqref{Low_DeltaY_deviation} to the time-varying permissible output deviation set $ \Delta\mathcal{Y}_i(\delta_i^*(k_0))$ (details in Section \ref{sec_nom_traj_error_propagation}). Similarly, the difference between the nominal terminal state and the $\mathbf{C}_0$-optimal terminal state is constrained to the time-varying permissible state deviation set $ \Delta\mathcal{Z}_i(\delta_i^*(k_0))$ in \eqref{Low_terminal}. Finally, \eqref{Low_IC} provides flexibility in initial condition $z_i(k|k)$ based on the RPI set computed by $\mathbf{C}_0$.

\begin{figure}[t]
    \centering
    \includegraphics[width=100mm]{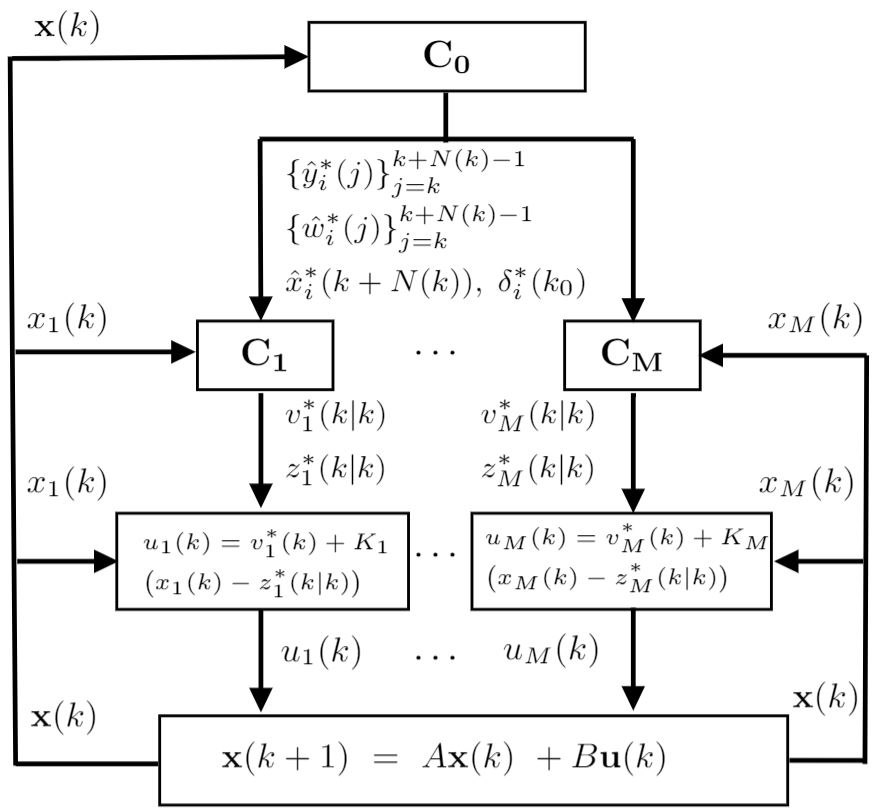}
\caption{Two-level hierarchical MPC where $\mathbf{C}_0$ is formulated based on \eqref{Up_MPC} and $\mathbf{C}_i$, $i \in \mathcal{N}$, based on \eqref{Low_MPC}. The $\mathbf{C}_0$-optimal trajectories $\hat{y}_i^*(j)$ and $\hat{w}_i^*(j)$ are computed using \eqref{upsamp_trajs} and \eqref{nominal_trajs}. The optimal output deviations $\delta_i^*(k_0)$ are used to coordinate controllers $\mathbf{C}_0$ and $\mathbf{C}_i$, $i \in \mathcal{N}$, and the static feedback control law \eqref{Control_Law} computes the inputs to each subsystem $\mathbf{S}_i$.}
\label{Fig_Hierarchy_BlkDiagram}
\end{figure}

As shown in Fig. \ref{Fig_Hierarchy_BlkDiagram}, coordination between the upper-level controller $\mathbf{C}_0$ and lower-level controllers $ \mathbf{C}_i $, $ i \in \mathcal{N} $, is achieved
through the communication of the $\mathbf{C}_0$-optimal trajectories $\hat{y}_i^*(j)$ and $\hat{w}_i^*(j)$, $j \in [k, k + N(k) -1]$, terminal state $\hat{x}_i^*(k+N(k))$, and the time-varying permissible deviation vectors $ \delta_i^*(k_0) $. In this hierarchical control architecture, only the lower-level controllers $ \mathbf{C}_i $ directly affect the system through inputs to the subsystems $\mathbf{S}_i$. Once each $ \mathbf{C}_i $
has solved for the optimal nominal input trajectories $ V_i^*(k) $ and optimal nominal initial condition $ z_i^*(k|k) $, the input to the system is $ \mathbf{u}(k) = [u_i(k)] $ where 
\begin{equation}\label{Control_Law}
	u_i(k) = v_i^*(k|k) + K_i(x_i(k) - z_i^*(k|k)).
\end{equation}
The two-level hierarchical controller is implemented based on \textbf{Algorithm 1}. The specific formulation of the sets in \eqref{Up_MPC} and \eqref{Low_MPC} are presented in Section \ref{sec_nom_traj_error_propagation} and the corresponding constraints are used to guarantee satisfaction of the state, input, and terminal constraints from \eqref{sys_cons} and \eqref{sys_Tcons} in Section \ref{sec_hier_control_feas}.

\IncMargin{1.5em}
\begin{algorithm2e}
	\SetAlgoLined
    \BlankLine
	Initialize $k$, $k_0 \leftarrow 0$ \\
	\If{$k \mod \nu_0 = 0$}{solve $\mathbf{P}_0(\mathbf{x}(k))$\;
			communicate $\{\hat{y}_i^*(j)\}_{j=k}^{k+N(k)-1}$, $\{\hat{w}_i^*(j)\}_{j=k}^{k+N(k)-1}$, $\hat{x}_i^*(k+N(k))$, and $\delta_i^*(k_0)$
			to $\mathbf{P}_i(x_i(k)), \forall i \in \mathcal{N} $\;
			$k_0 \leftarrow k_0 +1 $\;}
		solve $\mathbf{P}_i(x_i(k)), \forall i \in \mathcal{N}$, and apply the input $\mathbf{u}(k) = [u_i(k)]$ to the system based on \eqref{Control_Law}\; 
		$k \leftarrow k + 1$\;
	\caption{Two-level Hierarchical MPC with subsystem coupling.}
	\label{Hier_Algorithm}
\end{algorithm2e}
\DecMargin{1.5em}

\section{Nominal Trajectories and Error Propagation} \label{sec_nom_traj_error_propagation}

Following the tube-based MPC framework in \cite{Mayne2005}, the goal of this section is to explicitly bound the differences between the nominal state and input trajectories planned by the controllers $\mathbf{C}_0$ and $ \mathbf{C}_i $, $ i \in \mathcal{N} $, and the true system trajectories. 

First, since $ \mathbf{C}_0 $ has a larger time step size than $ \mathbf{C}_i$ and system dynamics (i.e. $ \Delta t_0 > \Delta t $), the input and state trajectories determined by $ \mathbf{C}_0 $ must be upsampled. Let $ \hat{\mathbf{u}}^*(k) $ and $ \hat{\mathbf{x}}^*(k) $ be the upsampled input and state trajectories corresponding to the optimal trajectories determined by $ \mathbf{C}_0 $.  Since the model \eqref{Up_model} assumed a piecewise constant input, the upsampled trajectories are computed as the forward simulation of \eqref{sys_dynamics} such that
\begin{subequations}\label{upsamp_trajs}
\begin{align}
	\hat{\mathbf{u}}^*(k) &= \hat{\mathbf{u}}^*(k_0|k_0),\\
	\hat{\mathbf{x}}^*(k) &= A^{k-\nu_0 k_0} \hat{\mathbf{x}}^*(k_0|k_0) + \mkern-20mu \sum_{j=0}^{k-\nu_0 k_0 -1} \mkern-20mu A^{j}B \hat{\mathbf{u}}^*(k_0|k_0),
\end{align}
\end{subequations}
for $ k \in [\nu_0 k_0, \nu_0 (k_0 + 1) - 1] $. These trajectories create the $\mathbf{C}_0$-optimal output and disturbance trajectories $\hat{y}_i^*(k)$ and $\hat{w}_i^*(k)$ used in \eqref{Low_DeltaY_deviation} and \eqref{Low_model}, where $\hat{u}_i^*(k) = \Pi_i \; \hat{\mathbf{u}}^*(k)$, $\hat{x}_i^*(k) = \Pi_i \; \hat{\mathbf{x}}^*(k)$, and
\begin{subequations} \label{nominal_trajs}
\begin{align}
    \hat{y}_i^*(k) &= [ \hat{x}_i^*(k)^\top \;\hat{u}_i^*(k)^\top ]^\top, \label{nominal_output} \\
    \hat{w}_i^*(k) &= \sum_{j \in \mathcal{N}_i} (A_{ij}\hat{x}_j^*(k) + B_{ij}\hat{u}_j^*(k)). \label{nominal_Distrubance}
\end{align}
\end{subequations}
Having defined the upsampled nominal trajectories for $ \mathbf{C}_0 $, let $ \Delta \mathbf{x}(k) = [\Delta x_i(k)] $, $ \Delta \mathbf{u}(k) = [\Delta u_i(k)] $, and $ \Delta \mathbf{y}(k) = [\Delta y_i(k)] $ denote the state, input, and output prediction errors for $ \mathbf{C}_0 $, where
\begin{subequations}\label{Up_Prediction_Error_Defn}
\begin{align}
    \Delta x_i(k) &\triangleq x_i(k) - \hat{x}_i^*(k), \\
    \Delta u_i(k) &\triangleq u_i(k) - \hat{u}_i^*(k), \\
    \Delta y_i(k) &\triangleq y_i(k) - \hat{y}_i^*(k) = [\Delta x_i(k)^\top \, \Delta u_i(k)^\top]^\top. 
\end{align}
\end{subequations} 
As shown in Fig. \ref{Fig_Hierarchy_errors}, these upper-level prediction errors consist of two parts, corresponding to the planned deviations by lower-level controllers $ \mathbf{C}_i $ and the resulting lower-level prediction errors due to the coupling between subsystems. Specifically,
\begin{subequations} \label{Up_Prediction_Error}
\begin{align}
	\Delta x_i(k) &= \Delta z_i(k) + e_i(k),\\
	\Delta u_i(k) &= \Delta v_i(k) + K_i e_i(k),
\end{align}
\end{subequations}
where 
\begin{equation}\label{subsys_deviations}
    \Delta z_i(k) \triangleq z_i(k) - \hat{x}_i^*(k), \; \Delta v_i(k) \triangleq v_i(k) - \hat{u}_i^*(k), 
\end{equation}
are the planned deviations and 
\begin{equation*}
    e_i(k) \triangleq x_i(k) - z_i(k), 
\end{equation*}
are lower-level prediction errors due to the coupling between subsystems. Note that $ K_i e_i(k) = u_i(k) - v_i(k) $ based on the control law from \eqref{Control_Law}. With the planned deviations bounded as $ \Delta z_i(k) \in \Delta \mathcal{Z}_i(k_0) $ and $ \Delta v_i(k) \in \Delta \mathcal{V}_i(k_0)$, which are simultaneously imposed as bounded output deviations in \eqref{Low_DeltaY_deviation}, the following two lemmas establish prediction error bounds for the lower- and upper-level controllers.

\begin{figure}
    \centering
    \includegraphics[width=100mm]{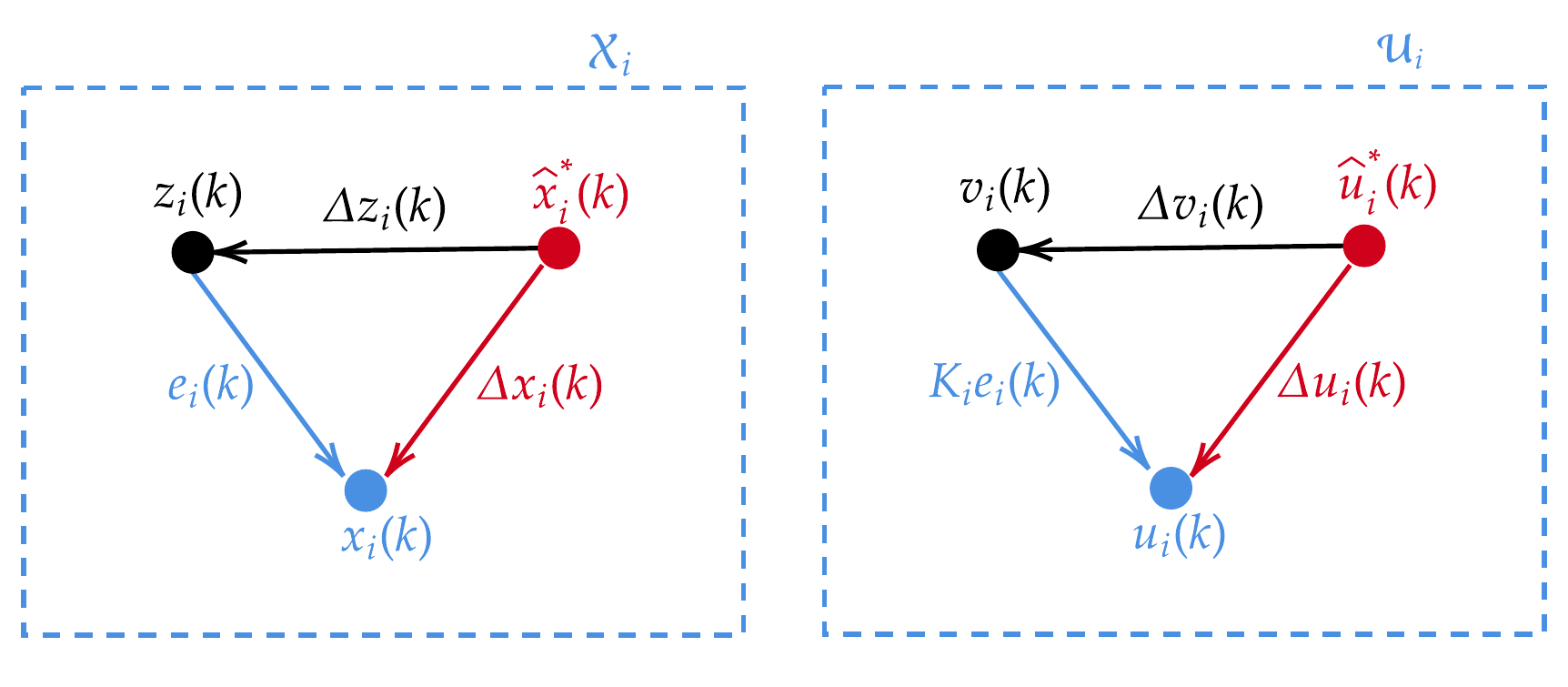}
    \caption{Schematic showing the true states and inputs for the system compared to those planned by controllers $\mathbf{C}_0$ and $\mathbf{C}_i$, $i \in \mathcal{N}$, and the corresponding planned and unplanned prediction errors.}
    \label{Fig_Hierarchy_errors}
\end{figure}

\begin{lemma} \label{lemma_Low_error_bounds}
    Let the disturbance error set be defined as
    \begin{equation}\label{Up_dist_set}
        \Delta\mathcal{W} = A_C\Delta \mathcal{Z} \oplus B_C \Delta \mathcal{V}.
    \end{equation}
    Then the lower-level prediction errors $ \mathbf{e}(k) = [e_i(k)] $ are bounded to the RPI set $ \mathcal{E}_0 \subset \mathbb{R}^n $, where $ \mathcal{E}_0 $ satisfies
    \begin{equation} \label{Up_RPI}
        (A + BK)\mathcal{E}_0 \oplus \Delta\mathcal{W} \subseteq \mathcal{E}_0. 
    \end{equation}
\end{lemma}
\begin{pf}
    Using the true subsystem dynamics from \eqref{subsys_dynamics} and the nominal subsystem model from \eqref{Low_model}, the lower-level prediction error dynamics for each subsystem are 
    \begin{equation} \label{subsystem_error_dynamics}
        e_i(k+1) = (A_{ii} + B_{ii} K_i) e_i(k) + \Delta w_i(k),
    \end{equation}
    where $\Delta w_i(k) = w_i(k) - \hat{w}_i^*(k)$. Using the definitions of $  w_i(k) $ and $ \hat{w}_i^*(k)$ from \eqref{disturbances} and \eqref{nominal_Distrubance}, this disturbance error is
    \begin{equation*}
        \Delta w_i(k) = \sum_{j \in \mathcal{N}_i} (A_{ij}(x_j(k) - \hat{x}_j^*(k)) + B_{ij}(u_j(k) - \hat{u}_j^*(k))).
    \end{equation*}
    Based on \eqref{Up_Prediction_Error_Defn} and \eqref{Up_Prediction_Error}, this disturbance error can be rewritten as
    \begin{equation} \label{distubance_error}
	    \Delta w_i(k) = \sum_{j \in \mathcal{N}_i} \begin{bmatrix} (A_{ij} + B_{ij} K_j) e_j(k) + \\ A_{ij} \Delta z_j(k) + B_{ij} \Delta v_j(k) \end{bmatrix}.
    \end{equation}
    Combining \eqref{subsystem_error_dynamics} and \eqref{distubance_error} for all subsystems $ i \in \mathcal{N} $ results in the system error dynamics
    \begin{equation*}
        \mathbf{e}(k+1) = (A + BK) \mathbf{e}(k) + A_C \Delta \mathbf{z}(k) + B_C \Delta \mathbf{v}(k).
    \end{equation*}
    Since $ \Delta z_i(k) \in \Delta \mathcal{Z}_i $ and $ \Delta v_i(k) \in \Delta \mathcal{V}_i, $ $ \forall i \in \mathcal{N} $, $ A_C \Delta \mathbf{z}(k) + B_C \Delta \mathbf{v}(k) \in \Delta\mathcal{W}$, as defined in \eqref{Up_dist_set}. Thus, if $ \mathbf{e}(k) \in \mathcal{E}_0 $ and $ \mathcal{E}_0 $ satisfies \eqref{Up_RPI}, then $ \mathbf{e}(k+1) \in \mathcal{E}_0 $.
    \hfill \hfill \qed
\end{pf}

\begin{lemma} \label{lemma_pred_error0}
    The upper-level prediction errors $ \Delta \mathbf{x}(k) = [\Delta x_i(k)] $ and $ \Delta \mathbf{u}(k) = [\Delta u_i(k)] $ are bounded such that 
    \begin{equation}
        \Delta \mathbf{x}(k) \in \Delta \mathcal{Z} \oplus \mathcal{E}_0, \quad \Delta \mathbf{u}(k) \in \Delta \mathcal{V} \oplus K \mathcal{E}_0.
    \end{equation}
\end{lemma}
\begin{pf}
    The proof follows directly from the definitions of $ \Delta x_i(k) $ and $ \Delta u_i(k) $ from \eqref{Up_Prediction_Error} and the result of \textbf{Lemma \ref{lemma_Low_error_bounds}}.
    \hfill \hfill \qed
\end{pf}

Based on the results of \textbf{Lemmas \ref{lemma_Low_error_bounds}} and \textbf{\ref{lemma_pred_error0}}, the nominal outputs determined by the upper-level controller in \eqref{Up_outputs} are constrained to the time-varying tightened output constraint set $ \hat{\mathcal{Y}}_0(\boldsymbol{\delta}(k_0))$. For notational simplicity, let $ \Delta \mathcal{Z} = \Delta \mathcal{Z} (\boldsymbol{\delta}(k_0)) $,  $ \Delta  \mathcal{V} = \Delta \mathcal{V} (\boldsymbol{\delta}(k_0)) $, and $ \mathcal{E}_0 = \mathcal{E}_0 (\boldsymbol{\delta}(k_0)) $. Then the time-varying tightened output constraint set is defined as 
\begin{equation}\label{Up_tight_op_const}
    \hat{\mathcal{Y}}_0(\boldsymbol{\delta}(k_0)) \triangleq \tilde{\mathcal{Y}}_0 \ominus [(\Delta \mathcal{Z} \oplus \mathcal{E}_0) \times (\Delta \mathcal{V} \oplus K \mathcal{E}_0 )],
\end{equation} 
where $ \tilde{\mathcal{Y}}_0 \subseteq \mathcal{Y}_0$ is a tightened output constraint set used to prevent inter-sample constraint violations (see \textbf{Appendix A.1} for details). 
Similarly, in \eqref{Up_terminal}, the nominal terminal state is constrained to the time-varying tightened terminal constraint set $ \hat{\mathcal{T}}_0(\boldsymbol{\delta} (k_0)) $ defined as
\begin{equation}\label{Up_tight_term_const}
    \hat{\mathcal{T}}_0(\boldsymbol{\delta} (k_0)) \triangleq \mathcal{T} \ominus \left(\Delta \mathcal{Z} \oplus \mathcal{E}_0\right).
\end{equation}
Note that $ \hat{\mathcal{Y}}_0(\boldsymbol{\delta}(k_0)) $ and $ \hat{\mathcal{T}}_0(\boldsymbol{\delta} (k_0)) $ are functions of the time-varying permissible output deviations $\boldsymbol{\delta}(k_0)$ due to the definition of $ \Delta \mathcal{Z} (\boldsymbol{\delta}(k_0))$ and $ \Delta \mathcal{V} (\boldsymbol{\delta}(k_0))$ and their direct impact on $ \mathcal{E}_0 (\boldsymbol{\delta}(k_0))$, as established in \textbf{Lemma \ref{lemma_Low_error_bounds}}. 

\section{Hierarchical Control Feasibility}\label{sec_hier_control_feas}
The following establishes recursive feasibility of each controller in the hierarchy and guarantees constraint satisfaction for the closed-loop system.

\begin{assume}\label{assm_IC}
    There exists a feasible solution to $\mathbf{P}_0(\mathbf{x}(0))$ at time step $k = k_0 = 0$ for the initial condition $\mathbf{x}(0)$.
\end{assume}
\begin{lemma}\label{lem_P0_feas_imp_Pi_feas}
If $\mathbf{P}_0(\mathbf{x}(k))$ is feasible time step $k = \nu_0k_0$, then $\mathbf{P}_i(x_i(k)), i \in \mathcal{N},$ is feasible at this time step.
\end{lemma}
\begin{proof}
Let $\{\hat{\mathbf{x}}^*(j)\}_{j = k}^{k+N(k)}$, $\{\hat{\mathbf{u}}^*(j)\}_{j = k}^{k+N(k)-1}$, and $\boldsymbol{\delta}^*(k_0)$ denote upsampled $\mathbf{C}_0$-optimal state and input trajectories and the optimal output deviation determined by $\mathbf{C}_0$ at time step $k = \nu_0k_0$. First, for all $i \in \mathcal{N}$, it is to be shown that there exists an initial condition $\mathbf{z}(k|k) = [z_i(k|k)]$ that simultaneously satisfies the output deviation constraints from \eqref{Low_DeltaY_deviation} and the initial condition constraint from \eqref{Low_IC}. Specifically, the initial condition constraint \eqref{Up_IC} ensures that $\mathbf{x}(k) - \hat{\mathbf{x}}^*(k_0|k_0) \in \Delta \mathcal{Z}(\boldsymbol{\delta}^*(k_0)) \oplus \mathcal{E}_0(\boldsymbol{\delta}^*(k_0))$. Based on the structure of these sets from \eqref{sys_perm_dev_sets_DeltaZ} and \eqref{Up_struct_RPI}, $x_i(k) - \hat{x}_i^*(k_0|k_0) \in \Delta \mathcal{Z}_i(\delta_i^*(k_0)) \oplus \mathcal{E}_i(\delta_i^*(k_0))$ for all $i \in \mathcal{N}$. Since these sets are all zonotopes, let $ \Delta\mathcal{Z}_i(\delta_i^*(k_0)) = \{G_i^z,0\} $ and $ \mathcal{E}_i(\delta_i^*(k_0)) = \{G_i^\varepsilon,0\} $. Therefore, $x_i(k) - \hat{x}_i^*(k_0|k_0) \in \Delta \mathcal{Z}_i(\delta_i^*(k_0)) \oplus \mathcal{E}_i(\delta_i^*(k_0))$ guarantees the existence of $ \xi_i^z $ and $ \xi_i^\varepsilon $ such that $ || \xi_i^z ||_\infty \leq 1 $, $ || \xi_i^\varepsilon ||_\infty \leq 1 $, and
\begin{equation} \label{lem3_x_k_defn}
    x_i(k) - \hat{x}_i^*(k_0|k_0) = G_i^z \xi_i^z + G_i^\varepsilon \xi_i^\varepsilon.
\end{equation}
Choosing $ z_i(k|k) = \hat{x}_i^*(k_0|k_0) + G_i^z \xi_i^z $, ensures that this initial condition satisfies the output deviation constraint from \eqref{Low_DeltaY_deviation}. Solving for $ \hat{x}_i^*(k_0|k_0) $ and plugging into \eqref{lem3_x_k_defn} results in $ x_i(k) - z_i(k|k) =  G_i^\varepsilon \xi_i^\varepsilon $ and thus this choice of initial condition also satisfies the initial condition constraint from \eqref{Low_IC}.

It remains to show the existence of a candidate solution starting from this initial condition $z_i(k|k)$, denoted by the nominal input sequence $\{v_i(j|k)\}_{j = k}^{k+N(k)-1}$ and corresponding nominal state sequence $\{z_i(j|k)\}_{j = k}^{k+N(k)}$, that satisfies the model \eqref{Low_model} and the constraints \eqref{Low_outputs}-\eqref{Low_terminal}. Comparing the candidate solution satisfying the nominal subsystem dynamics from \eqref{Low_model} and the upsampled $\mathbf{C}_0$-optimal trajectories satisfying the the true subsystem dynamics from \eqref{subsys_dynamics} results in $ z_i(j+1|k) - \hat{x}_i^*(j+1) = A_{ii} (z_i(j|k) - \hat{x}_i^*(j)) + B_{ii} (v_i(j|k) - \hat{u}_i^*(j)) $ for all $ j \in [k, k+N(k)-1]$. Since $ z_i(k|k) $ has already been shown to satisfy the output deviation constraint from \eqref{Low_DeltaY_deviation}, $ z_i(k|k) - \hat{x}_i^*(k) \in \Delta \mathcal{Z}_i(\delta_i^*(k_0)) $. From \eqref{Up_Pre} and the definition of the precursor set from \eqref{precursor}, this guarantees the existence of $ v_i(k|k) $ such that $ v_i(k|k) - \hat{u}_i^*(k) \in \Delta \mathcal{V}_i(\delta_i^*(k_0)) $ and $ z_i(k+1|k) - \hat{x}_i^*(k+1) \in \Delta \mathcal{Z}_i(\delta_i^*(k_0)) $. This process is repeated to show that the output deviation constraints from \eqref{Low_DeltaY_deviation} can be satisfied for all time steps and that the terminal constraint from \eqref{Low_terminal} is satisfied at the final time step. 
\end{proof}
\begin{lemma}\label{lem_Low_betw_C0_update}
For all $i \in \mathcal{N}$, if $\mathbf{P}_i(x_i(k))$ is feasible at time step $k$, where $k \; \text{mod} \;\nu_0 = 0$ (i.e at the time of $\mathbf{C}_0$ update), then $\mathbf{P}_i(x_i(k))$ is feasible at each time step $k+1$ through $k+N(k)-1$.
\end{lemma}
\begin{proof}
Let the feasible solution for $\mathbf{P}_i(x_i(k))$ at time step $k$ be defined by the optimal nominal input sequence $\{v_i^*(j|k)\}_{j = k}^{k+N(k)-1}$
and corresponding optimal nominal state sequence $\{z_i^*(j|k)\}_{j = k}^{k+N(k)}$ satisfying \eqref{Low_model}. While \eqref{Low_IC} guarantees $x_i(k) - z_i^*(k|k) \in \mathcal{E}_i(\delta_i^*(k_0))$, the feasibility of $\mathbf{P}_j(x_j(k))$, $j \in \mathcal{N} \setminus \{ i \} $, ensures that the disturbances $\Delta w_i$ from \eqref{distubance_error} due to subsystem coupling are bounded to $\Delta \mathcal{W}$ used to define $\mathcal{E}_0 = \mathcal{E}_1 \times \cdots \times \mathcal{E}_M$ in \textbf{Lemma \ref{lemma_Low_error_bounds}}. Thus $x_i(k+1) - z_i^*(k+1|k) \in \mathcal{E}_i(\delta_i^*(k_0))$,
and 
$\{z_i^*(j|k)\}_{j = k+1}^{k+N(k)}$ and $\{v_i^*(j|k)\}_{j = k+1}^{k+N(k)-1}$ are feasible nominal state and input sequences, which are the tails of sequences determined at previous time step $k$. Thus, $\mathbf{P}_i(x_i(k+1))$ is feasible and by induction, $\mathbf{P}_i(x_i(j))$, $\forall j \in [k+1, k+N(k)-1]$ is recursively feasible.
\end{proof}
\begin{lemma}\label{lem_Low_Up_update}
If $\mathbf{P}_i(x_i(k-1))$ $\forall i \in \mathcal{N}$ had feasible solutions at the previous time step $k-1$, then $\mathbf{P}_0(\mathbf{x}(k))$ has a feasible solution at current time step $k = \nu_0k_0$.
\end{lemma}
\begin{proof}
Let the candidate solution to $\mathbf{P}_0(\mathbf{x}(k))$ be the optimal nominal state and input sequences $\{\hat{\mathbf{x}}^*(j|k_0 -1)\}_{j = k_0}^{k_0+N_0(k_0)}$, $\{\hat{\mathbf{u}}^*(j|k_0 -1)\}_{j = k_0}^{k_0+N_0(k_0)-1}$, corresponding to the tails of the optimal solution determined at previous time step $k_0 -1$, and the previously planned deviation bound $\boldsymbol{\delta}^*(k_0-1)$. Since \eqref{Up_model}, \eqref{Up_outputs}, \eqref{Up_terminal}, and \eqref{Up_Pre} are time-invariant, the candidate solution satisfies these constraints. To show that $\hat{\mathbf{x}}^*(k_0|k_0-1)$ is a feasible initial condition, consider the following. Since, $\mathbf{P}_i(x_i(k-1))$ is feasible at time step $k-1$, the terminal state $z_i^*(k|k-1) = z_i^*(k_0)$
satisfies $z_i^*(k_0) - \hat{x}_i^*(k_0|k_0-1) \in \Delta \mathcal{Z}_i(\boldsymbol{\delta}_i^*(k_0-1))$ for every $\mathbf{S}_i$. Additionally, using the invariance of $\mathcal{E}_i(\delta_i^*(k_0-1))$ under control law \eqref{Control_Law}, $x_i(k) - z_i^*(k_0) \in \mathcal{E}_i(\delta_i^*(k_0-1))$. Thus, by combining these statements for all subsystems, $\hat{\mathbf{x}}^*(k_0|k_0-1)$ satisfies \eqref{Up_IC}. 
\end{proof}
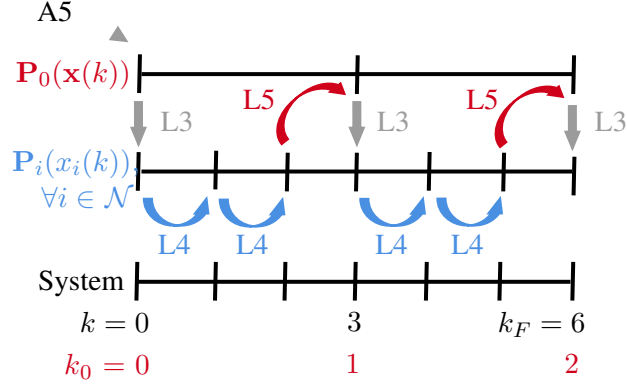
\begin{figure}
\centering
\tikzset{every picture/.style={line width=0.75pt}} 
\begin{tikzpicture}[x=0.75pt,y=0.75pt,yscale=-1,xscale=1]

\draw [line width=1.5]    (177.03,59) -- (392.03,59) ;
\draw [line width=1.5]    (176.27,49.06) -- (175.84,68.52) ;
\draw [line width=1.5]    (284.97,49.54) -- (284.55,69) ;
\draw [line width=1.5]    (392.96,49.54) -- (392.54,69) ;
\draw [line width=1.5]    (175.62,98.23) -- (175.2,117.69) ;
\draw [line width=1.5]    (284.55,98.54) -- (284.12,118) ;
\draw [line width=1.5]    (250.04,98.54) -- (249.61,118) ;
\draw [line width=1.5]    (214.32,97.54) -- (213.9,117) ;
\draw [line width=1.5]    (320.86,98.17) -- (320.44,117.63) ;
\draw [line width=1.5]    (358.02,98.54) -- (357.6,118) ;
\draw [line width=1.5]    (393.14,100.54) -- (392.71,120) ;
\draw  [color={rgb, 255:red, 74; green, 144; blue, 226 }  ,draw opacity=1 ][fill={rgb, 255:red, 74; green, 144; blue, 226 }  ,fill opacity=1 ] (192.86,135.21) .. controls (184.78,135.24) and (178.22,129.18) .. (178.2,121.69) -- (182.27,121.67) .. controls (182.29,129.17) and (188.85,135.22) .. (196.93,135.19) ;\draw  [color={rgb, 255:red, 74; green, 144; blue, 226 }  ,draw opacity=1 ][fill={rgb, 255:red, 74; green, 144; blue, 226 }  ,fill opacity=1 ] (196.93,135.19) .. controls (202,135.17) and (206.45,132.77) .. (209.07,129.13) -- (210.43,129.12) -- (209.49,121.57) -- (203.64,129.15) -- (205,129.14) .. controls (202.38,132.78) and (197.92,135.19) .. (192.86,135.21)(196.93,135.19) -- (192.86,135.21) ;
\draw [line width=1.5]    (174.51,163.46) -- (391.17,163.46) ;
\draw [line width=1.5]    (392.54,153.46) -- (392.11,172.92) ;
\draw [line width=1.5]    (356.14,153.46) -- (355.71,172.92) ;
\draw [line width=1.5]    (319.44,153.46) -- (319.01,172.92) ;
\draw [line width=1.5]    (283.95,153.46) -- (283.52,172.92) ;
\draw [line width=1.5]    (248.86,153.46) -- (248.44,172.92) ;
\draw [line width=1.5]    (214.15,153.46) -- (213.72,172.92) ;
\draw [line width=1.5]    (174.94,153.46) -- (174.51,172.92) ;
\draw [shift={(170.55,43)}, rotate = 212.07] [fill={rgb, 255:red, 155; green, 155; blue, 155 }  ,fill opacity=1 ][line width=0.08]  [draw opacity=0] (8.93,-4.29) -- (0,0) -- (8.93,4.29) -- cycle    ;
\draw  [color={rgb, 255:red, 155; green, 155; blue, 155 }  ,draw opacity=1 ][fill={rgb, 255:red, 155; green, 155; blue, 155 }  ,fill opacity=1 ] (393.86,75.05) -- (393.88,87.95) -- (395.63,87.94) -- (392.14,96.54) -- (388.63,87.95) -- (390.38,87.95) -- (390.36,75.06) -- cycle ;
\draw  [color={rgb, 255:red, 74; green, 144; blue, 226 }  ,draw opacity=1 ][fill={rgb, 255:red, 74; green, 144; blue, 226 }  ,fill opacity=1 ] (230.74,135.64) .. controls (222.66,135.67) and (216.09,129.61) .. (216.08,122.12) -- (220.15,122.1) .. controls (220.17,129.6) and (226.73,135.65) .. (234.81,135.62) ;\draw  [color={rgb, 255:red, 74; green, 144; blue, 226 }  ,draw opacity=1 ][fill={rgb, 255:red, 74; green, 144; blue, 226 }  ,fill opacity=1 ] (234.81,135.62) .. controls (239.88,135.6) and (244.33,133.2) .. (246.95,129.56) -- (248.31,129.55) -- (247.37,122) -- (241.52,129.58) -- (242.88,129.57) .. controls (240.26,133.21) and (235.8,135.62) .. (230.74,135.64)(234.81,135.62) -- (230.74,135.64) ;
\draw  [color={rgb, 255:red, 74; green, 144; blue, 226 }  ,draw opacity=1 ][fill={rgb, 255:red, 74; green, 144; blue, 226 }  ,fill opacity=1 ] (300.91,135.27) .. controls (292.83,135.29) and (286.27,129.24) .. (286.25,121.74) -- (290.32,121.73) .. controls (290.34,129.23) and (296.9,135.28) .. (304.98,135.25) ;\draw  [color={rgb, 255:red, 74; green, 144; blue, 226 }  ,draw opacity=1 ][fill={rgb, 255:red, 74; green, 144; blue, 226 }  ,fill opacity=1 ] (304.98,135.25) .. controls (310.05,135.23) and (314.51,132.83) .. (317.12,129.19) -- (318.48,129.18) -- (317.54,121.63) -- (311.69,129.21) -- (313.05,129.2) .. controls (310.43,132.84) and (305.98,135.25) .. (300.91,135.27)(304.98,135.25) -- (300.91,135.27) ;
\draw  [color={rgb, 255:red, 74; green, 144; blue, 226 }  ,draw opacity=1 ][fill={rgb, 255:red, 74; green, 144; blue, 226 }  ,fill opacity=1 ] (339.4,135.14) .. controls (331.32,135.17) and (324.76,129.11) .. (324.74,121.62) -- (328.81,121.6) .. controls (328.83,129.1) and (335.4,135.15) .. (343.48,135.12) ;\draw  [color={rgb, 255:red, 74; green, 144; blue, 226 }  ,draw opacity=1 ][fill={rgb, 255:red, 74; green, 144; blue, 226 }  ,fill opacity=1 ] (343.48,135.12) .. controls (348.54,135.1) and (353,132.7) .. (355.61,129.06) -- (356.97,129.05) -- (356.03,121.5) -- (350.18,129.08) -- (351.54,129.07) .. controls (348.93,132.71) and (344.47,135.12) .. (339.4,135.14)(343.48,135.12) -- (339.4,135.14) ;
\draw  [color={rgb, 255:red, 208; green, 2; blue, 27 }  ,draw opacity=1 ][fill={rgb, 255:red, 208; green, 2; blue, 27 }  ,fill opacity=1 ] (253.44,71.73) .. controls (245.21,78.35) and (242.35,88.42) .. (247.03,94.23) -- (250.19,91.69) .. controls (245.5,85.88) and (248.37,75.81) .. (256.6,69.19) ;\draw  [color={rgb, 255:red, 208; green, 2; blue, 27 }  ,draw opacity=1 ][fill={rgb, 255:red, 208; green, 2; blue, 27 }  ,fill opacity=1 ] (256.6,69.19) .. controls (262.7,64.27) and (270.05,62.64) .. (275.34,64.51) -- (276.39,63.66) -- (278.39,68.99) -- (271.13,67.9) -- (272.18,67.05) .. controls (266.89,65.19) and (259.54,66.81) .. (253.44,71.73)(256.6,69.19) -- (253.44,71.73) ;
\draw [line width=1.5]    (177.03,108) -- (392.03,108) ;
\draw  [color={rgb, 255:red, 208; green, 2; blue, 27 }  ,draw opacity=1 ][fill={rgb, 255:red, 208; green, 2; blue, 27 }  ,fill opacity=1 ] (362.44,73.22) .. controls (354.21,79.84) and (351.35,89.92) .. (356.03,95.72) -- (359.19,93.18) .. controls (354.5,87.37) and (357.37,77.3) .. (365.6,70.68) ;\draw  [color={rgb, 255:red, 208; green, 2; blue, 27 }  ,draw opacity=1 ][fill={rgb, 255:red, 208; green, 2; blue, 27 }  ,fill opacity=1 ] (365.6,70.68) .. controls (371.7,65.76) and (379.05,64.13) .. (384.34,66) -- (385.39,65.15) -- (387.39,70.48) -- (380.13,69.39) -- (381.18,68.54) .. controls (375.89,66.68) and (368.54,68.3) .. (362.44,73.22)(365.6,70.68) -- (362.44,73.22) ;
\draw  [color={rgb, 255:red, 155; green, 155; blue, 155 }  ,draw opacity=1 ][fill={rgb, 255:red, 155; green, 155; blue, 155 }  ,fill opacity=1 ] (286.27,73.05) -- (286.29,85.95) -- (288.04,85.94) -- (284.55,94.54) -- (281.04,85.95) -- (282.79,85.95) -- (282.77,73.06) -- cycle ;
\draw  [color={rgb, 255:red, 155; green, 155; blue, 155 }  ,draw opacity=1 ][fill={rgb, 255:red, 155; green, 155; blue, 155 }  ,fill opacity=1 ] (177.27,73.05) -- (177.29,85.95) -- (179.04,85.94) -- (175.55,94.54) -- (172.04,85.95) -- (173.79,85.95) -- (173.77,73.06) -- cycle ;

\draw (112.84,50.4) node [anchor=north west][inner sep=0.75pt]    {$\mathbf{\textcolor[rgb]{0.82,0.01,0.11}{P}}\textcolor[rgb]{0.82,0.01,0.11}{_{{0}}}\textcolor[rgb]{0.82,0.01,0.11}{\mathnormal{(}}\mathbf{\textcolor[rgb]{0.82,0.01,0.11}{x}}\mathnormal{\textcolor[rgb]{0.82,0.01,0.11}{(}\textcolor[rgb]{0.82,0.01,0.11}{k}\textcolor[rgb]{0.82,0.01,0.11}{)}\textcolor[rgb]{0.82,0.01,0.11}{)}}$};
\draw (104.34,93.4) node [anchor=north west][inner sep=0.75pt]    {$ \begin{array}{l}
\mathbf{\textcolor[rgb]{0.29,0.56,0.89}{P}\textcolor[rgb]{0.29,0.56,0.89}{_{\mathnormal{i}}}}\mathnormal{\textcolor[rgb]{0.29,0.56,0.89}{(}\textcolor[rgb]{0.29,0.56,0.89}{x}\textcolor[rgb]{0.29,0.56,0.89}{_{\mathnormal{i}}}\textcolor[rgb]{0.29,0.56,0.89}{(}\textcolor[rgb]{0.29,0.56,0.89}{k}\textcolor[rgb]{0.29,0.56,0.89}{)}\textcolor[rgb]{0.29,0.56,0.89}{)}}\textcolor[rgb]{0.29,0.56,0.89}{,}\\
\textcolor[rgb]{0.29,0.56,0.89}{\ \ \ \forall }\textcolor[rgb]{0.29,0.56,0.89}{\mathnormal{i\in }}\textcolor[rgb]{0.29,0.56,0.89}{\ \mathcal{N}}
\end{array}$};
\draw (144.75,177.32) node [anchor=north west][inner sep=0.75pt]    {$k = 0$};
\draw (278.26,177.32) node [anchor=north west][inner sep=0.75pt]    {$3$};
\draw (353.45,177.32) node [anchor=north west][inner sep=0.75pt]    {$k_{F} = 6$};
\draw (136.93,198.32) node [anchor=north west][inner sep=0.75pt]    {$\mathnormal{\textcolor[rgb]{0.82,0.01,0.11}{k}}\textcolor[rgb]{0.82,0.01,0.11}{_{0}}\ \textcolor[rgb]{0.82,0.01,0.11}{=}\textcolor[rgb]{0.82,0.01,0.11}{\ 0}$};
\draw (277.44,198.32) node [anchor=north west][inner sep=0.75pt]    {$\textcolor[rgb]{0.82,0.01,0.11}{1}$};
\draw (386.62,198.32) node [anchor=north west][inner sep=0.75pt]    {$\textcolor[rgb]{0.82,0.01,0.11}{2}$};
\draw (123.93,156) node [anchor=north west][inner sep=0.75pt]   [align=left] {{\fontfamily{ptm}\selectfont System}};
\draw (123.47,21) node [anchor=north west][inner sep=0.75pt]   [align=left] {{\fontfamily{ptm}\selectfont A5}};
\draw (337.74,65) node [anchor=north west][inner sep=0.75pt]   [align=left] {{\fontfamily{ptm}\selectfont \textcolor[rgb]{0.82,0.01,0.11}{L5}}};
\draw (226.25,65) node [anchor=north west][inner sep=0.75pt]   [align=left] {{\fontfamily{ptm}\selectfont \textcolor[rgb]{0.82,0.01,0.11}{L5}}};
\draw (401.86,76.05) node [anchor=north west][inner sep=0.75pt]   [align=left] {{\fontfamily{ptm}\selectfont \textcolor[rgb]{0.61,0.61,0.61}{L3}}};
\draw (183.9,138) node [anchor=north west][inner sep=0.75pt]  [color={rgb, 255:red, 74; green, 144; blue, 226 }  ,opacity=1 ] [align=left] {{\fontfamily{ptm}\selectfont \textcolor[rgb]{0.29,0.56,0.89}{L4}}};
\draw (222.74,138) node [anchor=north west][inner sep=0.75pt]  [color={rgb, 255:red, 74; green, 144; blue, 226 }  ,opacity=1 ] [align=left] {{\fontfamily{ptm}\selectfont \textcolor[rgb]{0.29,0.56,0.89}{L4}}};
\draw (292.32,138) node [anchor=north west][inner sep=0.75pt]  [color={rgb, 255:red, 74; green, 144; blue, 226 }  ,opacity=1 ] [align=left] {{\fontfamily{ptm}\selectfont \textcolor[rgb]{0.29,0.56,0.89}{L4}}};
\draw (330.18,138) node [anchor=north west][inner sep=0.75pt]  [color={rgb, 255:red, 74; green, 144; blue, 226 }  ,opacity=1 ] [align=left] {{\fontfamily{ptm}\selectfont \textcolor[rgb]{0.29,0.56,0.89}{L4}}};
\draw (293.19,76.05) node [anchor=north west][inner sep=0.75pt]   [align=left] {{\fontfamily{ptm}\selectfont \textcolor[rgb]{0.61,0.61,0.61}{L3}}};
\draw (185.19,76.05) node [anchor=north west][inner sep=0.75pt]   [align=left] {{\fontfamily{ptm}\selectfont \textcolor[rgb]{0.61,0.61,0.61}{L3}}};
\end{tikzpicture}
\vspace{-5pt}
\caption{Schematic showing how \textbf{Assumption \ref{assm_IC}} and \textbf{Lemmas \ref{lem_P0_feas_imp_Pi_feas}-\ref{lem_Low_Up_update}} are used to establish feasibility of two-level hierarchical controller with coupling between subsystems in the lower-level.}   \label{fig_hier_feasibility}
\end{figure}
\begin{theorem}
Following \emph{\textbf{Algorithm \ref{Hier_Algorithm}}} for a two-level hierarchical controller with $M$ controllers in the lower-level, all control problems, $\mathbf{P}_0(\mathbf{x}(k))$ and $\mathbf{P}_i(x_i(k))$, $\forall i \in \mathcal{N}$, are feasible,
resulting in system state and input trajectories satisfying state, input, and output constraints from \eqref{sys_cons} and terminal constraint from \eqref{sys_Tcons}.
\end{theorem}
\begin{proof}
Using \textbf{Assumption \ref{assm_IC}} and \textbf{Lemmas 3-5}, Fig. \ref{fig_hier_feasibility} shows how feasibility is established for $\mathbf{C}_0$ and $\mathbf{C}_i$, $\forall i \in \mathcal{N}$. For notational convenience, let \begin{equation*}
\Omega \triangleq \Delta \mathcal{Y}(\boldsymbol{\delta}^*(k_0)) \oplus \Big(\mathcal{E}_0(\boldsymbol{\delta}^*(k_0)) \times K\mathcal{E}_0(\boldsymbol{\delta}^*(k_0))\Big).    
\end{equation*}
Since, $\mathbf{P}_0(\mathbf{x}(k))$ is feasible, the output trajectory $\hat{\mathbf{y}}^*(j)$ satisfies
\begin{equation*}
    \hat{\mathbf{y}}^*(j) \in \hat{\mathcal{Y}}_0^*(k_0) \triangleq \tilde{\mathcal{Y}}_0 \ominus \Omega \subseteq \mathcal{Y} \ominus \Omega, 
\end{equation*}
based on \eqref{Up_outputs}, \eqref{Up_tight_op_const}, and the fact that $\tilde{\mathcal{Y}}_0 \subseteq \mathcal{Y}$. Similarly, the feasibility of $\mathbf{P}_i(x_i(k))$ guarantees that the output trajectory $y_i^*(k|k)$ generated by $\mathbf{C}_i$ satisfies 
\begin{subequations}
    \begin{align}
        & y_i^*(k|k) - \hat{y}_i^*(k) \in \Delta \mathcal{Y}_i (\delta_i^*(k_0)), \label{Thm1_subsys_op_dev_satisfaction} \\
        & y_i(k) - y_i^*(k|k) \in \mathcal{E}_i^*(\delta_i^*(k_0)) \times K_i\mathcal{E}_i(\delta_i^*(k_0)), \label{Thm1_subsys_op_const_satisfaction}
    \end{align}
\end{subequations}
based on \eqref{Low_DeltaY_deviation} and \eqref{Low_IC}.
Thus, adding \eqref{Thm1_subsys_op_dev_satisfaction} and \eqref{Thm1_subsys_op_const_satisfaction} for each system results in $\mathbf{y}(k) \in \hat{\mathbf{y}}^*(k) \oplus \Omega$.
Since $\hat{\mathbf{y}}^*(k) \in \mathcal{Y} \ominus \Omega$, $\mathbf{y}(k) \in (\mathcal{Y} \ominus \Omega) \oplus \Omega$. 
Finally, using the anti-extensive property of the set opening operation, $\mathbf{y}(k) \in  \mathcal{Y}$. Note that satisfaction of the terminal constraint from \eqref{sys_Tcons} can be proven similarly. 
\end{proof}

\section{Conclusions}\label{sec_conclusions}
A two-level hierarchical MPC formulation was presented for linear systems of dynamically-coupled subsystems. Adjustable tubes are used to bound permissible deviations between the system trajectories planned by the upper- and lower-level controllers. 
A tube-based robust MPC formulation with simultaneous uncertainty set optimization and constraint tightening 
guaranteed constraint satisfaction to bounded disturbances between subsystem controllers. This new approach 
extends the applicability of hierarchical control algorithms to system operation focusing on the notion of completion, where constraint feasible equilibrium might not exist. This document presented proofs that were omitted in the original publication due to space constraints. 

\section*{APPENDIX}
This appendix provides the details necessary to implement the proposed hierarchical MPC controller.  
\vspace{4mm}
\\
\textbf{A.1 Inter-sample tightened output constraint set computation for $\mathbf{C}_0$}

\noindent The tightened output constraint set $ \tilde{\mathcal{Y}}_0 \subseteq \mathcal{Y}_0$ introduced in \eqref{Up_tight_op_const} is used to ensure that the coarse trajectories planned by $\mathbf{C}_0$ produce upsampled trajectories in \eqref{upsamp_trajs} and \eqref{nominal_output} that satisfy the original output constraints. While there are many ways to achieve this, the approach used in this paper is based on bounding the difference between the upsampled trajectories and the linear interpolations of the coarse trajectories planned by $\mathbf{C}_0$. Specifically, let $\hat{\mathbf{x}}(k_0)$ and $\hat{\mathbf{x}}(k_0+1)$ denote the first two states predicted by $\mathbf{C}_0$ corresponding to the first input $\hat{\mathbf{u}}(k_0)$, where $\hat{\mathbf{x}}(k_0 +1) = A^{\nu_0}\hat{\mathbf{x}}(k_0) + \sum_{j = 0}^{\nu_0-1} A^jB \hat{\mathbf{u}}(k_0)$. The linearly interpolated trajectories $\mathbf{x}_l(k+i)$, $\forall i \in [1, \nu_0-1]$, can be computed between $\hat{\mathbf{x}}(k_0)$ and $\hat{\mathbf{x}}(k_0 +1)$ as 
\begin{equation*}
    \mathbf{x}_l(k+i) = \hat{\mathbf{x}}(k_0) + \frac{i}{\nu_0} (\hat{\mathbf{x}}(k_0+1) - \hat{\mathbf{x}}(k_0)).
\end{equation*}
By the convexity of $ \mathcal{X} $, $\hat{\mathbf{x}}(k_0), \hat{\mathbf{x}}(k_0+1) \in \mathcal{X}$ implies $\mathbf{x}_l(k+i) \in \mathcal{X}$. However, it is not guaranteed that the upsampled trajectory satisfies  $\hat{\mathbf{x}}(k+i) \in \mathcal{X}$, $\forall i \in [1, \nu_0-1]$. Defined as $\mathbf{e}(k+i) = \hat{\mathbf{x}}(k+i) - \mathbf{x}_l(k+i)$, the difference between these trajectories can be computed as
\begin{equation}
    \mathbf{e}(k+i) = A_e(i)\hat{\mathbf{x}}(k) +B_e(i)\hat{\mathbf{u}}(k),
\end{equation}
where $A_e(i) = A^i - \frac{i}{\nu_0}A^{\nu_0} -(1 -\frac{i}{\nu_0})I_n$ and $B_e(i) = \sum_{j = 0}^{i-1}A^jB - \frac{i}{\nu_0}\sum_{j = 0}^{\nu_0-1}A^jB$. Since $\hat{\mathbf{x}}(k) \in \mathcal{X}$ and $\hat{\mathbf{u}}(k) \in \mathcal{U}$, these differences are bounded such that $\mathbf{e}(k+i) \in \mathcal{F}_i = A_e(i)\mathcal{X} \oplus B_e(i)\mathcal{U}$. Thus, defining $\mathcal{F} = \mathcal{F}_1 \cup \cdots \cup \mathcal{F}_{\nu_0-1} $ ensures $ \mathbf{e}(k+i) \in \mathcal{F} $, $\forall i \in [1,\nu_0-1]$. Finally, computing $\tilde{\mathcal{Y}} = \tilde{\mathcal{X}} \times \mathcal{U}$ where $\tilde{\mathcal{X}} = \mathcal{X} \ominus \mathcal{F} $ guarantees that the upsampled trajectories satisfy the original constraints. For ease of implementation, outer-approximating bounding boxes of $\mathcal{F}_i$ are computed and thus, $\mathcal{F}$ is also a bounding box. 
\vspace{4mm}
\\
\textbf{A.2 RPI set computation}

\noindent This section presents how the RPI set $\mathcal{E}_0(\boldsymbol{\delta}(k_0))$, used in \eqref{Up_IC}, is computed through the addition of linear constraints and cost function terms in the formulation of $\mathbf{P}_0(\mathbf{x}(k))$ using the approach from \cite{Raghuraman_2021_ACC}. Before proceeding with the computation of $\mathcal{E}_0(\boldsymbol{\delta}(k_0))$, consider the following definition. 
\begin{defn}\label{def_zono_scaled} \emph{\cite{Raghuraman_2021_ACC}}
The zonotope $\mathcal{Z}(\Phi) = \{G\Phi, c \} \subset \mathbb{R}^n$ is a scaled version of the nominal zonotope $\mathcal{Z} = \{G, c\}$ with the generator matrix $G$ scaled by a diagonal matrix $\Phi \in \mathbb{R}^{n_g \times n_g}$, $\Phi = \text{diag}(\phi_i)$, $\phi_i \geq 0$, $\forall i \in [1,n_g]$.
\end{defn}
\noindent Let the RPI set be a scaled zonotope such that  $\mathcal{E}_0(\boldsymbol{\delta}) = \{G_{\varepsilon} \Phi_{\varepsilon}, \mathbf{0} \}$ with an \emph{a priori} chosen nominal generator matrix $ G_{\varepsilon} \in \mathbb{R}^{n \times n_\varepsilon} $. The permissible state and input deviation sets are defined as scaled zonotopes with centers at the origin such that $\Delta \mathcal{Z} = \{G_{z} \textit{diag}(\boldsymbol{\delta}^z), \mathbf{0} \}$ and $\Delta \mathcal{V} = \{G_{v} \textit{diag}(\boldsymbol{\delta}^v), \mathbf{0} \}$, where $ G_z \in \mathbb{R}^{n \times n_z} $ and $ G_v \in \mathbb{R}^{m \times n_v} $. From \eqref{Up_dist_set}, the resulting disturbance error set is a scaled zonotope such that $\Delta \mathcal{W} = \{G_{w} \textit{diag}(\boldsymbol{\delta}), \mathbf{0} \}$, where $G_w = [ A_C G_{z} \; B_C G_{v} ] \in \mathbb{R}^{n \times n_w} $ and $ n_w = n_z + n_v $. 

\noindent Following the approach from \cite{Raghuraman_2021_ACC}, based on the one-step RPI computation from \cite{Raghuraman2019} and the zonotope containment conditions from \cite{Sadraddini2019}, the decision variables 
$ \Phi_{\varepsilon} \in \mathbb{R}^{n_\varepsilon \times n_\varepsilon} $, $ \Gamma_{\varepsilon,1} \in \mathbb{R}^{n_\varepsilon \times n_\varepsilon} $, and $ \Gamma_{\varepsilon,2} \in \mathbb{R}^{n_\varepsilon \times n_w} $ are added to $\mathbf{P}_0(\mathbf{x}(k))$ with linear constraints
\begin{subequations} \label{Dist_RPI_set_comp_all}
\begin{align}
    &(A+BK)G_{\varepsilon}\Phi_{\varepsilon}  = G_{\varepsilon} \Gamma_{\varepsilon,1},\\
    &G_{w}\textit{diag}(\boldsymbol{\delta}) = G_{\varepsilon} \Gamma_{\varepsilon,2},\\
    & |\Gamma_{\varepsilon,1}|\mathbf{1} + |\Gamma_{\varepsilon,2}|\mathbf{1} \leq \Phi_{\varepsilon}\mathbf{1}.
\end{align}
\end{subequations}
The cost function for $\mathbf{P}_0(\mathbf{x}(k))$ is modified to balance system performance with the maximization of $\boldsymbol{\delta}$ through the addition of the term $ \Lambda ||\bar{\boldsymbol{\delta}} - \boldsymbol{\delta}||_p $, where $ \Lambda $ is a scalar weighting term and $ \bar{\boldsymbol{\delta}} $ is a user-specified upper-bound on $ \boldsymbol{\delta} $. 
For the linear constraints \eqref{Dist_RPI_set_comp_all} to admit a feasible solution, the generator matrix $G_{\varepsilon}$ needs to be chosen carefully. As in \cite{Raghuraman2019,Raghuraman_2021_ACC}, an intuitive choice of generators is based on $G_w$ and $\bar{A}_K = \textit{diag}(A_{ii}+B_{ii}K_i)$
such that 
\begin{equation}\label{Up_RPI_gen_struct}
G_{\varepsilon} = [G_{w} \; \bar{A}_KG_{w} \; \cdots \; \bar{A}_K^{n_s}G_{w}],    
\end{equation}
where $n_s \in \mathbb{Z}_{+}$ is a parameter that can be increased to promote the feasibility of \eqref{Dist_RPI_set_comp_all} at the cost of set complexity and the number of decision variables. 

\noindent Note that the block-diagonal structure of $\bar{A}_K$ and structure of $G_w$ ensures that $G_{\varepsilon}$ from \eqref{Up_RPI_gen_struct} is separable and thus, $\mathcal{E}_0$ is a structured RPI set satisfying \eqref{Up_struct_RPI}. Once $\mathcal{E}_0(\boldsymbol{\delta}^*(k_0)) = \{G_{\varepsilon}\Phi_{\varepsilon}, c_{\varepsilon}\}$ is computed, the subsystem-level RPI sets $\mathcal{E}_i = \{G_{\varepsilon}^i, c_{\varepsilon}^i\}$, $ i \in \mathcal{N}$, satisfying \eqref{Up_struct_RPI} can be obtained by projection. 
\newpage
\noindent \textbf{A.3 Output Constraint Tightening for $\mathbf{C}_0$}

\noindent This section presents how the output constraint tightening used to compute $\hat{\mathcal{Y}}_0(\boldsymbol{\delta}(k_0))$, based on \eqref{Up_tight_op_const} and used in \eqref{Up_outputs}, is integrated into the formulation of $\mathbf{P}_0(\mathbf{x}(k))$ through the addition of linear constraints and cost function terms.

\noindent In \eqref{Up_tight_op_const}, let the inter-sample tightened output constraint set be a zonotope such that $\tilde{\mathcal{Y}}_0 = \{\tilde{G}_y, \tilde{c}_y \}$ with known generator matrix $ \tilde{G}_y \in \mathbb{R}^{(n+m) \times n_{\tilde{y}}} $ and center $ \tilde{c}_y \in \mathbb{R}^{n+m}$. Let the tightened output constraint set be a scaled zonotope such that $ \hat{\mathcal{Y}}_0(\boldsymbol{\delta}(k_0)) = \{\hat{G}_y \Phi_y, \hat{c}_y \} $, where $ \hat{G}_y \in \mathbb{R}^{(n+m) \times n_{\hat{y}}} $ is an \emph{a priori} chosen nominal generator matrix and $ \Phi_y $ is a diagonal scaling matrix satisfying $\Phi_y = \textit{diag}(\phi_y)$ with $\phi_{y,i} \geq 0$, $\forall i \in [1, n_{\hat{y}}]$.

\noindent Following the approach from \cite{Raghuraman_2021_ACC}, based on the one-step Pontryagin difference computation from \cite{Raghuraman2019} and the zonotope containment conditions from \cite{Sadraddini2019}, the decision variables $ \hat{c}_y \in \mathbb{R}^{n+m} $, $ \Phi_y \in \mathbb{R}^{n_{\hat{y}} \times n_{\hat{y}}} $, $ \Gamma_{y} \in \mathbb{R}^{n_{\tilde{y}} \times (n_{\hat{y}} + n_w + 2 n_\varepsilon)} $, and $ \beta_y \in \mathbb{R}^{n_{\tilde{y}}} $ are added to $\mathbf{P}_0(\mathbf{x}(k))$ with linear constraints

\begin{subequations} \label{Tight_opconsset_comp_all}
\begin{align}
    & \begin{bmatrix} \hat{G}_y\Phi_y \; \scriptsize  \begin{bmatrix} [ G_{z} \textit{diag}(\boldsymbol{\delta}^z) \; G_{\varepsilon}\Phi_{\varepsilon}] & \mathbf{0}
    \\ \mathbf{0} & [ G_{v} \textit{diag}(\boldsymbol{\delta}^v) \; KG_{\varepsilon}\Phi_{\varepsilon}] \end{bmatrix} 
\end{bmatrix} 
    = \tilde{G}_y\Gamma_y, \\
    &\tilde{c}_y - \hat{c}_y  = \tilde{G}_y\beta_y,\\
    &|\Gamma_y|\mathbf{1} + |\beta_y| \leq \mathbf{1}.
\end{align}
\end{subequations}
The cost function for $\mathbf{P}_0(\mathbf{x}(k))$ is modified to balance system performance with maximizing the size of $ \hat{\mathcal{Y}}_0(\boldsymbol{\delta}(k_0)) $ through the addition of the term $ -||\phi_y||_p $, where $ \phi_y $ is the vector of scaling variables along the diagonal of $ \Phi_y $. 
Note that the choice of $ \hat{G}_y $ affects the quality of the inner-approximation of the Pontryagin difference from \eqref{Up_tight_op_const}. 

\vspace{4mm}
\noindent \textbf{A.4 Terminal Constraint Tightening for $\mathbf{C}_0$}

\noindent The terminal constraint tightening used to compute $\hat{\mathcal{T}}_0(\boldsymbol{\delta}^*(k_0))$, based on \eqref{Up_tight_term_const} and used in \eqref{Up_terminal}, is integrated into the formulation of $\mathbf{P}_0(\mathbf{x}(k))$ through the addition of linear constraints and cost function terms following the same approach used in the previous section for output constraint tightening, and thus is not repeated here for brevity.

\vspace{4mm}
\noindent \textbf{A.5 Set containment condition ($\Delta \mathcal{Z} \subseteq \text{Pre} (\Delta \mathcal{Z})$)}

\noindent This section presents how the set containment 
$\Delta \mathcal{Z} \subseteq \text{Pre} (\Delta \mathcal{Z})$ in \eqref{Up_Pre} can be enforced using linear constraints. Assuming the invertibility of $ A_D $, the precursor set is a zonotope defined as $ \text{Pre} (\Delta \mathcal{Z}) = \{G_{p} \textit{diag}(\boldsymbol{\delta}), \mathbf{0} \} $, with $ G_{p} = [A_D^{-1} G_z \; -A_D^{-1} B_D G_v ] $. Using the zonotope containment conditions from \cite{Sadraddini2019}, the decision variable $ \Gamma_{p} \in \mathbb{R}^{n_w \times n_z} $ is added to $\mathbf{P}_0(\mathbf{x}(k))$ with linear constraints
\begin{subequations}
\begin{align}
    & G_z \textit{diag}(\boldsymbol{\delta}^z) = G_p \Gamma_{p}, \\
    & |\Gamma_p|\mathbf{1} \leq \textit{diag}(\boldsymbol{\delta}) \mathbf{1}.
\end{align}
\end{subequations}

\bibliography{library}

\begin{thebibliography}{10}

\bibitem{Wang2017a}
Ye~Wang, Vicen{\c{c}} Puig, and Gabriela Cembrano.
\newblock {Non-linear economic model predictive control of water distribution
  networks}.
\newblock {\em Journal of Process Control}, 56:23--34, 2017.

\bibitem{Seok2017a}
Jinwoo Seok, Ilya Kolmanovsky, and Anouck Girard.
\newblock {Coordinated model predictive control of aircraft gas turbine engine
  and power system}.
\newblock {\em Journal of Guidance, Control, and Dynamics}, 40(10):2538--2555,
  2017.

\bibitem{Khan2016}
Irfan Khan, Zhicheng Li, Yinliang Xu, and Wei Gu.
\newblock {Distributed control algorithm for optimal reactive power control in
  power grids}.
\newblock {\em International Journal of Electrical Power and Energy Systems},
  83:505--513, 2016.

\bibitem{Irfan2017}
Muhammad Irfan, Jamshed Iqbal, Adeel Iqbal, Zahid Iqbal, Raja~Ali Riaz, and
  Adeel Mehmood.
\newblock {Opportunities and challenges in control of smart grids – Pakistani
  perspective}.
\newblock {\em Renewable and Sustainable Energy Reviews}, 71 (2017):652--674,
  2017.

\bibitem{Enang2017}
Wisdom Enang and Chris Bannister.
\newblock {Modelling and control of hybrid electric vehicles (A comprehensive
  review)}.
\newblock {\em Renewable and Sustainable Energy Reviews}, 74:1210--1239, 2017.

\bibitem{Wang2020_DSCC}
Wenqing Wang and Justin~P Koeln.
\newblock {Hierarchical Multi-Timescale Energy Management for Hybrid-Electric
  Aircraft}.
\newblock {\em ASME Dynamic Systems and Control Conference}, 2020.

\bibitem{Scattolini2009}
Riccardo Scattolini.
\newblock {Architectures for distributed and hierarchical Model Predictive
  Control - A review}.
\newblock {\em Journal of Process Control}, 19:723--731, 2009.

\bibitem{Koeln2019_Aut}
Justin~P. Koeln, Vignesh Raghuraman, and Brandon~M. Hencey.
\newblock {Vertical hierarchical MPC for constrained linear systems}.
\newblock {\em Automatica}, 113:108817, 2020.

\bibitem{Barcelli2010}
Davide Barcelli, Alberto Bemporad, and Giulio Ripaccioli.
\newblock {Hierarchical multi-rate control design for constrained linear
  systems}.
\newblock {\em Proceedings of the IEEE Conference on Decision and Control},
  pages 5216--5221, 2010.

\bibitem{Farina2017b}
M.~Farina, X.~Zhang, and R.~Scattolini.
\newblock {A hierarchical MPC scheme for interconnected systems}.
\newblock {\em IFAC-PapersOnLine}, 50(1):12021--12026, 2017.

\bibitem{Farina2018}
Marcello Farina, X.~Zhang, and Riccardo Scattolini.
\newblock {A hierarchical multi-rate MPC scheme for inter-connected systems}.
\newblock {\em Automatica}, 90:38--46, 2018.

\bibitem{Vermillion2014}
Chris Vermillion, Amor Menezes, and Ilya Kolmanovsky.
\newblock {Stable hierarchical model predictive control using an inner loop
  reference model and $\lambda$-contractive terminal constraint sets}.
\newblock {\em Automatica}, 50(1), 2014.

\bibitem{Sampathnarayanan2014}
Balaji Sampathnarayanan, Simona Onori, and Stephen Yurkovich.
\newblock {An optimal regulation strategy with disturbance rejection for energy
  management of hybrid electric vehicles}.
\newblock {\em Automatica}, 50:128--140, 2014.

\bibitem{doman2015rapid}
David~B Doman.
\newblock Rapid mission planning for aircraft thermal management.
\newblock {\em AIAA Guidance, Navigation, and Control Conference}, page 1076,
  2015.

\bibitem{Richards2003}
Arthur Richards and Jonathan~P. How.
\newblock {Model Predictive Control of Vehicle Maneuvers with Guaranteed
  Completion Time and Robust Feasibility}.
\newblock {\em American Control Conference}, pages 4034--4040, 2003.

\bibitem{Raghuraman2020_ACC}
Vignesh Raghuraman, Venkatraman Renganathan, Tyler~H. Summers, and Justin~P.
  Koeln.
\newblock {Hierarchical MPC with Coordinating Terminal Costs}.
\newblock {\em American Control Conference}, pages 4126--4133, 2020.

\bibitem{Raghuraman2019}
Vignesh Raghuraman and Justin~P. Koeln.
\newblock {Set operations and order reductions for constrained zonotopes}.
\newblock {\em arXiv:2009.06039v1}, 2020.

\bibitem{Raghuraman_2021_ACC}
Vignesh Raghuraman and Justin~P. Koeln.
\newblock {Tube-based robust MPC with adjustable uncertainty sets using
  zonotopes}.
\newblock {\em American Control Conference}, pages 462--469, 2021.

\bibitem{Koeln2019_ACC}
Justin~P. Koeln and Brandon~M. Hencey.
\newblock {Constrained Hierarchical MPC via Zonotopic Waysets}.
\newblock {\em American Control Conference}, pages 4237--4244, 2019.

\bibitem{Scattolini2007}
Riccardo Scattolini and Patrizio Colaneri.
\newblock {Hierarchical model predictive control}.
\newblock {\em Proceedings of the IEEE Conference on Decision and Control},
  pages 4803--4808, 2007.

\bibitem{Mayne2005}
D.Q. Mayne, M.M. Seron, and S.V. Rakovi{\'{c}}.
\newblock {Robust Model Predictive Control of Constrained Linear Systems with
  Bounded Disturbances}.
\newblock {\em Automatica}, 41:219--224, 2005.

\bibitem{Borrelli2011}
F.~Borrelli, A.~Bemporad, and M.~Morari.
\newblock {Predictive Control for Linear and Hybrid Systems}.
\newblock {\em Cambridge University Press}, 2011.

\bibitem{Sadraddini2019}
Sadra Sadraddini and Russ Tedrake.
\newblock {Linear Encodings for Polytope Containment Problems}.
\newblock {\em Proceedings of the IEEE Conference on Decision and Control},
  pages 4367--4372, 2019.

\end{thebibliography}
\bibliographystyle{unsrt}

\end{document}